\documentclass[sigconf, nonacm]{acmart}

\newcommand\vldbdoi{XX.XX/XXX.XX}
\newcommand\vldbpages{XXX-XXX}
\newcommand\vldbvolume{16}
\newcommand\vldbissue{2}
\newcommand\vldbyear{2022}
\newcommand\vldbauthors{\authors}
\newcommand\vldbtitle{\shorttitle} 
\newcommand\vldbavailabilityurl{https://github.com/Texera/Fries-Flink}
\newcommand\vldbpagestyle{plain}

\newcommand{\boldstart}[1]{\vspace{0.3em} \noindent \textbf{#1}}

\usepackage{amsmath}

\usepackage{xspace}

\usepackage{listings}
\usepackage{tikz}
\usetikzlibrary{decorations.pathreplacing,calc,shapes,positioning,tikzmark}

\usepackage[T1]{fontenc}
\usepackage[scaled=.85]{beramono}

\usepackage{enumitem}

\usepackage{xcolor}

\newcommand{\sysname}{{\sf Fries}\xspace}
\newcommand{\epoch}{{\sf Epoch}\xspace}
\usepackage{algorithm}
\usepackage[noend]{algpseudocode}

\usepackage{multirow}
\usepackage{subfig}

\usepackage{color}
\usepackage{array}
\usepackage{graphicx}

\newif\ifshort
\shorttrue

\begin{document}

\title{Fries: Fast and Consistent Runtime Reconfiguration in Dataflow Systems with Transactional Guarantees (Extended Version)}

\author{Zuozhi Wang, Shengquan Ni, Avinash Kumar, Chen Li}
\affiliation{%
  \institution{UC Irvine}
  \country{United States}
}
\email{{zuozhiw, shengqun, avinask1, chenli}@ics.uci.edu}

\begin{abstract}
A computing job in a big data system can take a long time to run, especially for pipelined executions on data streams. Developers often need to change the computing logic of the job such as fixing a loophole in an operator or changing the machine learning model in an operator with a cheaper model to handle a sudden increase of the data-ingestion rate.  Recently many systems have started supporting  runtime reconfigurations to allow this type of change on the fly without killing and restarting the execution. While the delay in reconfiguration is critical to performance, existing systems use epochs to do runtime reconfigurations, which can cause a long delay. In this paper we develop a new technique called \sysname that leverages the emerging availability of fast control messages in many systems, since these messages can be sent without being blocked by data messages.  We formally define consistency in runtime reconfigurations, and develop a \sysname scheduler with consistency guarantees. The technique not only works for different classes of dataflows, but also works for parallel executions and supports fault tolerance. Our extensive experimental evaluation on clusters show the advantages of this technique compared to epoch-based schedulers.
\end{abstract}

\maketitle

\pagestyle{\vldbpagestyle}
\begingroup\small\noindent\raggedright\textbf{PVLDB Reference Format:}\\
\vldbauthors. \vldbtitle. PVLDB, \vldbvolume(\vldbissue): \vldbpages, \vldbyear.\\
\href{https://doi.org/\vldbdoi}{doi:\vldbdoi}
\endgroup
\begingroup
\renewcommand\thefootnote{}\footnote{\noindent
This work is licensed under the Creative Commons BY-NC-ND 4.0 International License. Visit \url{https://creativecommons.org/licenses/by-nc-nd/4.0/} to view a copy of this license. For any use beyond those covered by this license, obtain permission by emailing \href{mailto:info@vldb.org}{info@vldb.org}. Copyright is held by the owner/author(s). Publication rights licensed to the VLDB Endowment. \\
\raggedright Proceedings of the VLDB Endowment, Vol. \vldbvolume, No. \vldbissue\ %
ISSN 2150-8097. \\
\href{https://doi.org/\vldbdoi}{doi:\vldbdoi} \\
}\addtocounter{footnote}{-1}\endgroup

\ifdefempty{\vldbavailabilityurl}{}{
\vspace{.3cm}
\begingroup\small\noindent\raggedright\textbf{PVLDB Artifact Availability:}\\
The source code, data, and/or other artifacts have been made available at \url{\vldbavailabilityurl}.
\endgroup
}

\definecolor{dkgreen}{rgb}{0,0.6,0}
\definecolor{gray}{rgb}{0.5,0.5,0.5}
\definecolor{mauve}{rgb}{0.58,0,0.82}

\lstdefinestyle{myStyle}{
  frame=tb,
  language=java,
  morekeywords={def, break, in, val, var},
  aboveskip=3mm,
  belowskip=3mm,
  showstringspaces=false,
  columns=flexible,
  basicstyle={\small\ttfamily},
  numbers=none,
  numberstyle=\tiny\color{gray},
  keywordstyle=\color{blue},
  commentstyle=\color{dkgreen},
  stringstyle=\color{mauve},
  frame=single,
  breaklines=true,
  breakatwhitespace=true,
  tabsize=2,
}

\section{Introduction}
\label{sec:introduction}

Big data systems are widely used to process large amounts of data. Each computation job in these systems can take a long time to run, from hours to days or even weeks to finish. 
Applications that require timely processing of input data often use pipelined dataflow execution engines~\cite{ journals/debu/CarboneKEMHT15, journals/pvldb/ChandramouliGBDPTW14, conf/sigmod/ToshniwalTSRPKJGFDBMR14, journals/pvldb/AkidauBCCFLMMPS15}, for example, in the scenarios of processing real-time streaming data, or answering queries progressively to provide early results to users. 
In these applications, when a long running job continuously processes ingested data, developers often need to change the computing logic of the job without disrupting the execution, as illustrated in the following example.  

Consider a data-processing pipeline for payment-fraud detection shown in Figure~\ref{fig:fd-reconfig-multi-op}. This simplified dataflow resembles many real-world applications~\cite{StreamINGFraudDetection,FlinkFraudDetectionDemo}. A stream of payment tuples is continuously ingested into the dataflow, with each tuple containing payment information such as customer, merchant, and amount. The dataflow uses two machine learning (ML) operators $FC$ and $FM$ to detect fraud based on customer and merchant information.  

\begin{figure}[htbp]
	\includegraphics[width=0.9\linewidth]{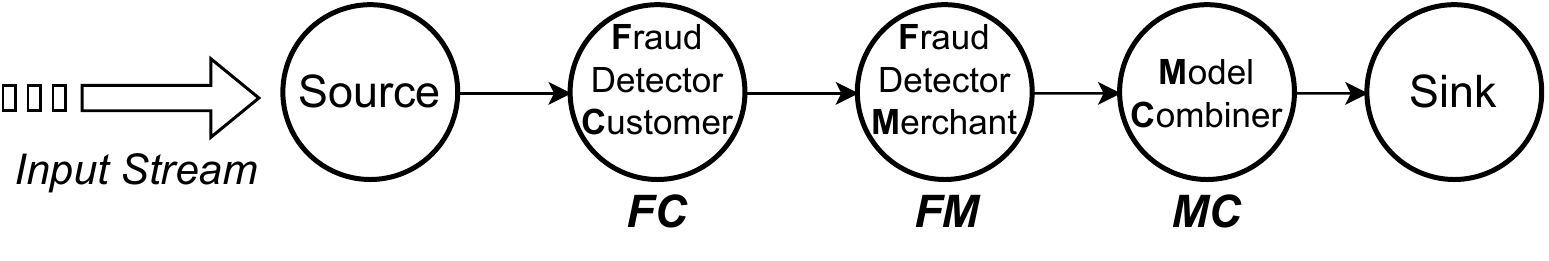}
	\caption{\label{fig:fd-reconfig-multi-op}
		\textbf{An example data-processing pipeline for fraud detection processing continuously ingested data. }
	}
\end{figure}

Consider two example use cases in this dataflow. {\em Use case 1: fixing loopholes in operators.} After observing unexpected tuples from the {\sf Sink} operator, the user identifies a loophole in the operator $FM$. She wants to update this operator to incorporate new rules to fix the loophole, without stopping the execution.  {\em Use case 2: handling surges of data arrival rate.} Suppose the data arrival rate at the source suddenly increases, and as a result, the end-to-end processing latency becomes larger. The user finds that the ML operator $FM$ is the bottleneck. To reduce the latency, she wants to ``hot-replace'' the expensive ML model (e.g., a deep neural network) with a lightweight model (e.g., a decision tree) to improve its performance, thus reduce the processing latency.  Again, she wants to make the change without stopping the execution.  These examples show the importance of allowing developers to change the dataflow execution ``on the fly.'' We call such changes {\em runtime reconfigurations}. This problem has gained a lot of interest in the research areas of software engineering~\cite{journals/tosem/SadeghiEM17}, mobile computing~\cite{journals/eis/KakousisPP10, conf/wmcsa/TiwariRMGL19}, and distributed systems~\cite{conf/coots/KonC99, conf/sigsoft/MaBGML11}. Recently, users of dataflow systems also show the need for runtime reconfigurations~\cite{StreamINGFraudDetection, FlinkUpdateLogicVol2, FlinkUpdateCepPattern} and  more systems start supporting this important feature~\cite{conf/sigmod/CarboneFKK20}, such as Amber~\cite{journals/pvldb/KumarWNL20}, Chi~\cite{journals/pvldb/MaiZPXSVCKMKDR18}, Flink~\cite{UpgradeFlinkApplications}, and Trisk~\cite{conf/cloud/MaoHTWM21}. 

Naturally there is a delay from the time a user requests a reconfiguration to the time its changes take effect in the target operators.  This delay is critical to the performance of the system. For example, in use case 1, the user wants to fix the loophole as soon as possible since a large reconfiguration delay can cause financial losses. In use case 2, a large delay in mitigating the surge can cause the system to suffer longer in terms of long latency and wasting of computing resources. Thus we want this delay to be as low as possible.

A main limitation of existing systems supporting runtime reconfigurations is that they could have a long reconfiguration delay. In these systems, after a reconfiguration request is submitted, they need to wait for all the in-flight tuples to be processed by those target reconfiguration operators, as well as those earlier operators in the dataflow, before the requested changes can be applied on the target operators. This delay could be very long, when there are many in-flight tuples, or some of these operators are expensive, especially for operators using advanced machine learning models and those implemented as user-defined functions (UDF's).

In this paper, we develop a novel technique, called ``\sysname,'' to perform runtime reconfigurations with a low delay.
It leverages the emerging availability of fast control messages in many systems recently. A {\em fast control message}, ``FCM'' for short, is a message exchanged between the controller in the data engine and an operator without being blocked by data messages.  Figure~\ref{fig:fd-reconfig-multi-op-dcm} shows an example of handling a reconfiguration request of two operators $FM$ and $MC$ using FCM's. Upon a reconfiguration request, the controller sends an FCM to each of the two operators, and each of them applies the new configuration immediately after receiving the message.  Since FCM's are sent separately from data messages, these changes can reach the target operators much faster. 

\begin{figure}[htbp]
	\includegraphics[width=0.85\linewidth]{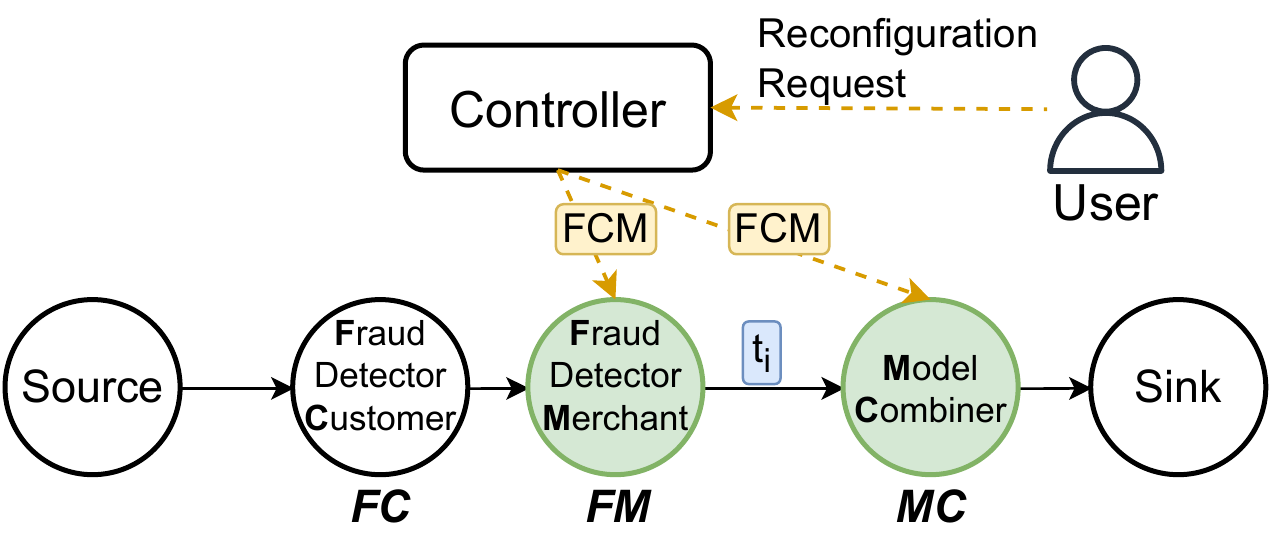}
	\caption{\label{fig:fd-reconfig-multi-op-dcm}
		\textbf{Handing a runtime reconfiguration of operators $FM$ and $MC$ using fast control messages (FCM's).}
	}
\end{figure}

We will show in Section~\ref{subsec:naive-fcm-scheduler} that the naive way of using FCM's  can cause consistency issues in Figure~\ref{fig:fd-reconfig-multi-op-dcm}.  It has unexpected side effects, e.g., producing incorrect results on the output tuples, or even causing the operator $MC$ to crash.  This example shows several challenges in developing \sysname: 1) What is the meaning of ``consistency'' in this reconfiguration context? 2) How to ensure this consistency  while reducing their delay?  3) How to deal with different types of operators and support parallel executions?  We study these challenges and make the following contributions.

\begin{itemize}
    \item We analyze epoch-based reconfiguration schedulers and show their limitations (Section~\ref{sec:epoch-reconfiguration}). 
    \item We formally define consistency of a reconfiguration based on transactions  (Section~\ref{sec:dcm-reconfiguration}).
    \item We first consider a simple class of dataflows that have one-to-one operators only, and develop a \sysname scheduler that guarantees consistency (Section~\ref{sec:hybrid-single-tuple}).
    \item We then consider the general class of dataflows with one-to-many operators, and extend the \sysname scheduler (Section~\ref{sec:complex-extension}). 
    \item We extend \sysname to more general cases, such as dataflows with blocking operators and multiple workers. We also discuss how to support fault-tolerance (Section~\ref{sec:extensions}).
    \item We conduct an extensive experimental study to evaluate \sysname in various scenarios and show its superiority compared to epoch-based schedulers (Section~\ref{sec:experiments}).
\end{itemize}

\subsection{Related Work}

\boldstart{Reconfiguration systems.} 
Recently, many data-processing systems have started to support reconfigurations. 
Flink~\cite{journals/debu/CarboneKEMHT15} supports reconfiguration by taking a savepoint~\cite{FlinkSavepoint}, killing the running job, then restarting the job with the new configuration. This approach is disruptive to the dataflow execution. 
Spark Streaming~\cite{conf/sigmod/ArmbrustDTYZX0S18, conf/sosp/ZahariaDLHSS13} uses a mini-batch-based execution strategy and supports reconfiguration between mini-batches. 
Chi~\cite{journals/pvldb/MaiZPXSVCKMKDR18} enables runtime reconfiguration by propagating epoch markers in its data stream. 
Trisk~\cite{conf/cloud/MaoHTWM21} provides an easy-to-use programming API for reconfigurations. 
The approaches in these systems are all based on epochs, which can have a long reconfiguration delay, as analyzed in Section~\ref{sec:epoch-reconfiguration}. \sysname relies on FCM's to perform reconfigurations with a low delay. 
Noria~\cite{conf/osdi/GjengsetSBAEKKM18} is a system that uses dataflows to incrementally maintain materialized views. The system supports reconfigurations of view definitions, which require the new views to be recomputed from entire base tables. In \sysname, an update of a dataflow only affects the future tuples. The input tuples that are already processed by the dataflow are not recomputed using the new configuration.

\boldstart{Re-scaling systems.} Some systems~\cite{conf/sigmod/FernandezMKP13,journals/pvldb/HoffmannLMKLR19,conf/sigmod/MonteZRM20} support updating the dataflow for re-scaling. For example, Megaphone~\cite{journals/pvldb/HoffmannLMKLR19} based on timely dataflow~\cite{conf/sosp/MurrayMIIBA13} supports a fine-granularity re-scaling and Rhino~\cite{conf/sigmod/MonteZRM20} based on Flink supports re-scaling with very large states. \sysname focuses on reconfiguring the computation functions of operators, which is different from re-scaling.

\boldstart{Transactions in dataflow systems.}
S-Store~\cite{journals/pvldb/MeehanTZACDKMMP15} and the work in~\cite{conf/edbt/BotanFKT12} are systems that allow streaming dataflows and OLTP workloads to access a shared mutable state.
Although both systems do not support reconfigurations directly, we could map a reconfiguration to these systems. S-Store defines transactions on the processing of each input batch on a single operator. This model cannot express our consistency requirements in reconfigurations.  The work in~\cite{conf/edbt/BotanFKT12} treats a dataflow as a black box, thus it has the limitation of not being able to utilize the properties of the dataflow and its operators to reduce the reconfiguration delay. \sysname can do so to achieve this reduction. Both earlier systems include a transaction scheduler to manage the processing of data, which creates scheduling overhead even when there is no reconfiguration.  The \sysname scheduler has no such overhead before receiving a reconfiguration request.
Additionally, both earlier systems are only on a single node, while the \sysname scheduler can run on a distributed engine on a cluster.

\boldstart{Transactions in database systems.}
Transactions are widely studied in traditional database systems (e.g., \cite{books/aw/BernsteinHG87,books/mk/WeikumV2002,books/mk/BernsteinN96}). 
A uniqueness in transactions in our work is that they treat operations in a reconfiguration as a separate transaction, which is handled differently from data transactions.  In addition, \sysname does optimizations by utilizing special properties in our problem setting, including the DAG shape of a dataflow, and types of operators, e.g., one-to-one and one-to-many.
Moreover, the \sysname scheduler uses FCM's and epoch markers to schedule transactions without  locking.

\section{Problem Settings}
\label{sec:preliminaries}

\subsection{Data-Processing Model}
\label{subsec:data-processing-model}

A data-processing system runs a computation dataflow job represented as a directed acyclic graph (DAG) of operators. Each operator receives tuples from its input edges, processes them, and sends tuples through its output edges. An operator contains a computation function $f$ represented as
$$f: (s, t) \rightarrow \big(s', \{(t'_1, o'_1), \ldots, (t'_n, o'_n)\}\big) .$$ 
The function processes a tuple $t$ at a time with a state $s$ of the operator, produces a set of zero or more output tuples $\{t'_1, \ldots, t'_n\}$, where each tuple $t'_i$ has a receiving operator $o'_i$. The operator also updates its state to $s'$.  
The system has a module called {\em controller} that manages the execution of the job, handles requests from the user, and exchanges messages with operators during the execution.

For simplicity, we first focus on dataflows under the following assumptions. (1) A dataflow contains pipelined operators only, such as selection, projection, union, and other tuple-at-a-time operators. We consider a class of join operators where the operator first collects all the tuples from one input (e.g., the ``build'' input of a hash join), then starts processing tuples from the other input (e.g., the ``probe'' phase of a hash join). We consider the processing of tuples from the second input of join. (2) Each operator has a single worker. We relax these assumptions in Section~\ref{sec:extensions}.

As an example, consider a data-processing pipeline for payment-fraud detection shown in Figure~\ref{fig:fd-reconfig-multi-op}.
The example dataflow uses two machine learning (ML) operators for fraud detection. The first one, denoted as $FC$, keeps a state of the 5 recent tuples of each customer. For each input tuple, $FC$ updates the state and feeds the 5 recent tuples of the customer into an ML model. The predicted probability $p_c(5)$ is attached as a new column of the tuple. The second one, denoted as $FM$, keeps a state of the 5 recent tuples of each merchant. Similarly, it uses an ML model to generate a predicted probability $p_m(5)$, and attaches it as a new column of the tuple. Finally, the model combiner $MC$ uses $p_c(5)$ and $p_m(5)$ of each tuple to compute the final average probability with the weights $[0.4, 0.6]$.

\subsection{Runtime Reconfiguration}
\label{subsec:runtime-reconfiguration}

\begin{definition}[Runtime reconfiguration]
\label{def:runtime-reconfiguration}
During the execution of a dataflow, an update to the computation functions of its operators is a {\em runtime reconfiguration} of this execution. 
\end{definition}

Formally, a reconfiguration $\mathcal{R}$ is a set of operators with a function update $\mu(o_i)$ for each operator $o_i$, i.e.,
$$ \mathcal{R} = \{ (o_1,\mu(o_1)),\ldots, (o_n,\mu(o_n)) \}. $$
Each operator $o_i$ has a {\em function-update operation} $\mu(o_i)$.  This operation applies a pair $\left<f_{o_i}', \mathcal{T}_{o_i}\right>$ to the operator, where $f_{o_i}'$ is a new computation function of the operator. $\mathcal{T}_{o_i}$ is a state transformation that converts the operator's original state $s$ to a new state $s^* = \mathcal{T}_{o_i}(s)$, which can be consumed by $f_{o_i}'$. In this paper, we consider the case where there is one reconfiguration at a time. 

In the running example, suppose the user identifies a flaw in the dataflow and wants to reconfigure the two operators $FM$ and $MC$. Specifically, the user wants to change $FM$ to output an additional probability value $p_m(10)$, which is predicted using the $10$ recent tuples of each merchant. The operator $MC$ needs to be updated to combine all three probabilities ($p_c(5)$, $p_m(10)$, and $p_m(5)$) with the new weights $[0.4, 0.4, 0.2]$. Table~\ref{table:fd-reconfig-multi-op-summary} shows the old and new configurations of the two reconfiguration operators.

\begin{table}[htbp]
\centering
\resizebox{2.5in}{!}{%
\begin{tabular}{|c|c|c|c|} 
\hline
    & $FM$'s output & $MC$ weights  \\ 
\hline
Old configuration & $p_m(5)$                       & {[}0.4, 0.6]                  \\ 
\hline
New configuration & $p_m(10)$, $p_m(5)$             & {[}0.4, 0.4, 0.2]        \\
\hline
\end{tabular}
}
\caption{Operator executions during a reconfiguration.}
\label{table:fd-reconfig-multi-op-summary}
\end{table}

Note that the new configuration of an operator can require a state different from that of the old configuration. In this case, the reconfiguration can use a state transformation to migrate the old state to the new one.
In the running example, the old configuration of operator $FM$ uses a state with the last $5$ payment tuples for each merchant. However, the new configuration of $FM$ needs a list of last $10$ tuples for each merchant. The user provides a state transformation $\mathcal{T}$ for operator $FM$, to instruct the system in transferring operator $FM$'s old state to the new one. In this example, the user chooses to fill the new state with the $5$ tuples from the old state and $5$ additional $null$ values.

\section{Epoch-Based Reconfiguration Schedulers and Limitations}
\label{sec:epoch-reconfiguration}

In this section, we explain epoch-based reconfiguration schedulers and show their limitation of long delays.

\subsection{Epoch-Based Schedulers}

\boldstart{Dataflow epoch.} A stream of tuples processed by the system can be divided into consecutive sets of tuples, where each set is called an {\em epoch}~\cite{journals/pvldb/CarboneEFHRT17}.  
One way to create epochs is to use epoch markers. At the start of a new epoch, an epoch marker is injected to each source operator. 
The epoch marker is then propagated along the data stream using the following protocol~\cite{journals/pvldb/CarboneEFHRT17}. 
When an operator receives an epoch marker from an input channel, it performs epoch alignment by waiting for all its inputs to receive an epoch marker, then sends the marker downstream. 
As an example, Figure~\ref{fig:epoch-fraud-propagate} shows two epochs during the execution of the fraud-detection dataflow. An epoch marker injected between $t_4$ and $t_5$ divides the input stream into two epochs. The epoch marker indicates the end of epoch 1 and the start of epoch 2.

\begin{figure}[htbp]
	\includegraphics[width=0.8\linewidth]{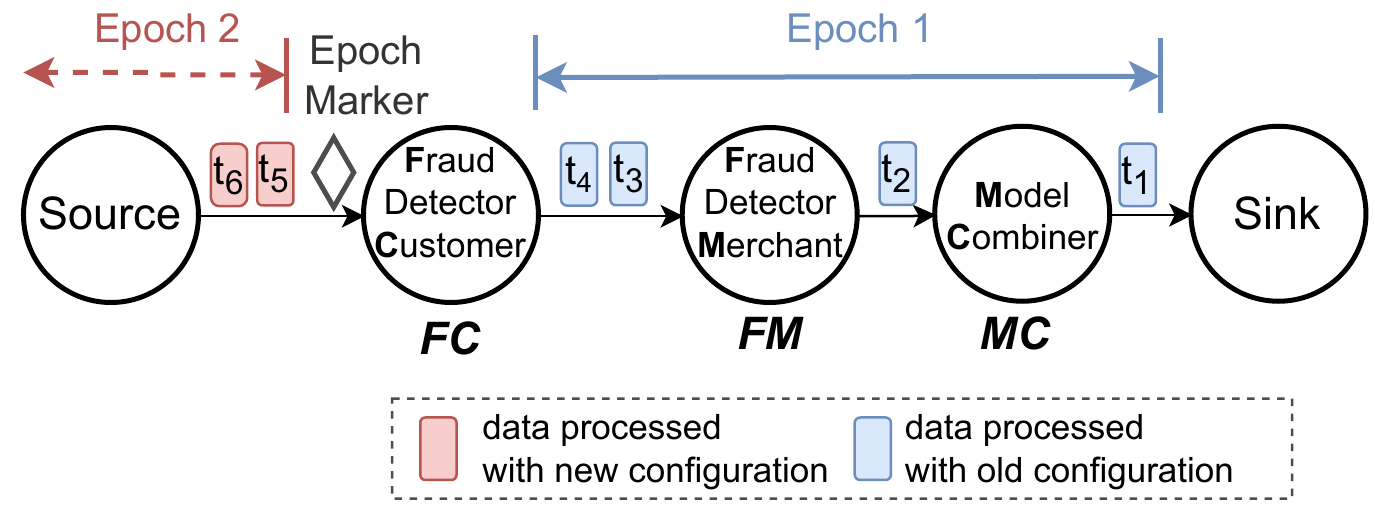}
	\caption{\label{fig:epoch-fraud-propagate}
		\textbf{An epoch-based reconfiguration scheduler in Chi~\cite{journals/pvldb/MaiZPXSVCKMKDR18}. It uses an epoch barrier to apply the new configuration to operators $FM$ and $MC$ at the start of Epoch 2.}
	}
\end{figure}

\begin{definition}[Epoch-based Scheduler]
An {\em epoch-based scheduler} schedules a reconfiguration request between two epochs. That is, for each reconfiguration operator $O$, all the tuples in the old epoch are processed with the old configuration of $O$, and all the tuples in the new epoch are processed with the new configuration of $O$.
\end{definition}

Considering the aforementioned method to generate epochs, the following is an implementation adopted by Chi~\cite{journals/pvldb/MaiZPXSVCKMKDR18}.  We call this implementation ``Epoch Barrier Reconfiguration'' scheduler, or ``EBR'' in short. 
Upon a reconfiguration request, 
the controller starts a new epoch and piggybacks the reconfiguration in the epoch marker.
When a reconfiguration operator receives epoch markers from all its inputs, it applies the new configuration.
The operator then processes the input tuples in the next epoch using the new configuration. Figure~\ref{fig:epoch-fraud-propagate} shows the process of handling a reconfiguration of operators $FM$ and $MC$ using the EBR scheduler.
When operator $FM$ receives the epoch marker, it applies the new configuration, and propagates the marker to operator $MC$. When operator $MC$ receives the epoch marker, it also applies the new configuration. 

\begin{table}[htbp]
\centering
\resizebox{\linewidth}{!}{%
\begin{tabular}{|>{\hspace{0pt}}m{0.300\linewidth}|>{\hspace{0pt}}m{0.270\linewidth}|>{\hspace{0pt}}m{0.430\linewidth}|} \hline
\textbf{System}                                                                                                                           & \textcolor[rgb]{0.2,0.2,0.2}{\textbf{Epoch creation}} & \textcolor[rgb]{0.2,0.2,0.2}{\textbf{Reconfiguration Strategy}}                                                             \\ \hline
\textcolor[rgb]{0.2,0.2,0.2}{Chi}~\cite{journals/pvldb/MaiZPXSVCKMKDR18} & \textcolor[rgb]{0.2,0.2,0.2}{Epoch makers}            & \textcolor[rgb]{0.2,0.2,0.2}{Piggybacking control }\textcolor[rgb]{0.2,0.2,0.2}{messages in epoch markers}        \\ \hline
\textcolor[rgb]{0.2,0.2,0.2}{Flink}~\cite{FlinkSavepoint}                                                                                  & \textcolor[rgb]{0.2,0.2,0.2}{Epoch makers}            & \textcolor[rgb]{0.2,0.2,0.2}{Stop-and-restart}                                                                     \\ \hline
\textcolor[rgb]{0.2,0.2,0.2}{Spark Streaming}~\cite{conf/sigmod/ArmbrustDTYZX0S18}                                                         & \textcolor[rgb]{0.2,0.2,0.2}{Mini-batch}              & \textcolor[rgb]{0.2,0.2,0.2}{Stop-and-restart}                                                                     \\ \hline
\end{tabular}
}
\caption{Epoch-based reconfiguration schedulers.}
\label{table:epoch-implementations}
\end{table}

Table~\ref{table:epoch-implementations} shows different epoch-based reconfiguration schedulers. %

In Flink~\cite{FlinkSavepoint}, upon a reconfiguration, a new epoch is immediately triggered using an epoch marker. At the end of the old epoch, each operator saves its state into a checkpoint. After the old epoch is processed by all the operators, Flink kills the execution, updates the dataflow graph, loads the saved states, and restarts the dataflow. In Spark Streaming~\cite{conf/sigmod/ArmbrustDTYZX0S18}, epochs are created by dividing the input data stream into small mini-batches, each of which is an epoch. A mini-batch is processed one at a time by launching a separate computation job. Upon a reconfiguration, the system modifies the dataflow graph before starting the job of the next mini-batch.

\subsection{Limitations: Long Reconfiguration Delays}
\label{subsec:comparison-latency}

A major limitation of epoch-based reconfiguration schedulers is a long reconfiguration delay, which is from the time a request is submitted to the time the new configuration takes effect in the target operators. In particular, the system needs to process all the in-flight tuples before the new epoch. 
Take the EBR scheduler in Figure~\ref{fig:epoch-fraud-propagate} as an example. Operator $FM$ needs to finish processing the in-flight tuples $t_3$ and $t_4$.
In general, this delay could be long due to the following reasons. First, the dataflow can contain multiple expensive operators that make the processing of an epoch slow. 
Second, the number of in-flight tuples could be large, especially when the system is under high workload. 
We may want to reduce the number of in-flight tuples by decreasing the buffer size.
However, a smaller buffer can be easily filled by a minor fluctuation in the input ingestion rate.
When the buffer is full, the system triggers back-pressure, which can decrease the throughput. 
Moreover, a small buffer size causes the networking layer to transmit data in small batches, which introduces additional transmission overhead. 
Compared to the EBR approach, the Flink approach suffers from an additional delay of stopping and restarting the dataflow. Spark Streaming can also have a long reconfiguration delay. The delay is determined by the processing time of a mini-batch, with a predefined interval usually set to a few seconds. However, the delay can be higher when the processing speed cannot keep up with a surge of the data ingestion rate.

\section{Scheduling Reconfigurations Using Fast Control Messages}
\label{sec:dcm-reconfiguration}

In this section, we introduce a new type of reconfiguration schedulers based on fast control messages (FCM's). We present a naive scheduler and show its issues. We then formally define consistency of a reconfiguration.

\begin{definition}[Fast Control Message]
A {\em fast control message}, ``FCM'' for short, is a message exchanged between the controller and an operator without being blocked by data messages.
\end{definition}

There are many ways to implement fast control messages. For instance, to send an FCM from the controller to the fraud detector in our running example, one approach is to set up a new communication channel between the controller and the fraud detector. The channel is separate from existing data channels, and the FCM can bypass data messages. Another way is to transmit the FCM using existing data channels, but assigning a higher priority to the FCM. The FCM is first sent to a source operator of the workflow, then propagated along the edges to the fraud detector, and it bypasses data messages in each data channel.

\subsection{FCM-based Schedulers}
\label{subsec:naive-fcm-scheduler}

\boldstart{Naive FCM scheduler.} A main benefit of using FCM's to schedule reconfigurations compared to epoch-based schedulers is that FCM's have a much smaller delay.  A naive scheduler leverages this benefit as follows. The controller sends an FCM directly to each reconfiguration operator. When an operator receives an FCM, it applies the new configuration immediately after finishing the processing of its current tuple.
We use Figure~\ref{fig:fd-reconfig-multi-op-dcm} to explain how the naive scheduler works in a reconfiguration of two operators $FM$ and $MC$.  Using this scheduler, the controller sends an FCM directly to each of the two operators $FM$ and $MC$. The FCM carries the new function $f'$ and the state transformation $\mathcal{T}$ of the corresponding operator. These operators update their configuration after receiving their FCM.

While this naive scheduler has a low reconfiguration delay, it could generate an undesirable reconfiguration schedule in this example.  Notice that the scheduler does not coordinate the updates to these two operators that run independently. Consider the in-flight tuple $t_i$, which is processed by $FM$ using its old configuration. Suppose the $MC$ switches to the new configuration before the arrival of $t_i$. Then tuple $t_i$ is processed by $MC$ using its new configuration. The tuple contains two probability values $p_c(5)$ and $p_m(5)$, but the new configuration of $MC$ expects three probability values.  This schema mismatch could have unexpected side effects, such as an incorrect result on the produced output tuple, or even causing the operator $MC$ to crash.  This example shows the importance for the reconfiguration to be performed in a synchronized manner on the two reconfiguration operators.  In particular, we want a tuple to be processed by the two operators either using the old configuration or using the new configuration.

\boldstart{FCM multi-version scheduler.}
To ensure a tuple is processed by the same configuration of multiple operators, 
we can use the following FCM-based multi-version scheduler that maintains multiple configurations of an operator at the same time. The controller first sends an FCM to each reconfiguration operator. Each operator keeps both the old configuration and the new one. After all operators have received the FCM, each source operator increments its version number, which is tagged to each source tuple.
For each input tuple, an operator checks the tuple's tagged version number, chooses the corresponding configuration version to process the tuple, and tags the same version number to the output tuples. As an example, in Figure~\ref{fig:fcm-multi-version}, after the new configuration is sent to operator $E$, the source operators then tag subsequent output tuples $t_3$ and $t_4$ with the new version $v_2$.

\begin{figure}[htbp]
	\includegraphics[width=2.0in]{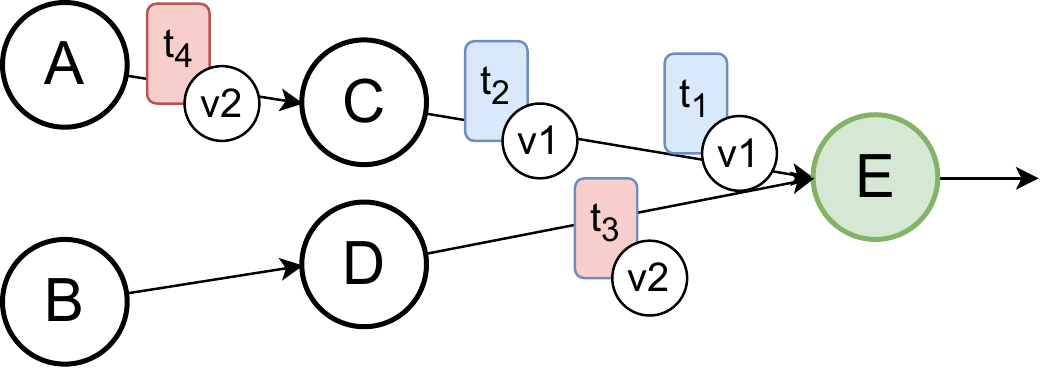}
	\caption{\label{fig:fcm-multi-version}
		\textbf{Using an FCM multi-version scheduler, an operator processes a tuple based on its version tag.}
	}
\end{figure}

This scheduler has two problems. First, each reconfigured operator may need to keep two sets of states for two configurations, and these states could be very large (e.g., large hash tables or machine learning models). Second, this scheduler still suffers from a possible high reconfiguration delay. In particular, similar to the case of the EBR scheduler, there can be a large amount of in-flight tuples that are already tagged with the old version and they still need to be processed with the old configuration (e.g., $t_1$ and $t_2$ in Figure~\ref{fig:fcm-multi-version}).

\subsection{Reconfiguration Consistency}
\label{subsec:reconfig-consistency}

We formally define the consistency requirements in this context.  At a high level, we treat the processing of a single source tuple by multiple operators as one {\em transaction}, and a reconfiguration as another transaction. We use conflict-serializability to define the consistency of a schedule of a reconfiguration.

\begin{definition}[Scope of a source tuple]
\label{def:tuple-scope}
The {\em scope} of a source tuple $t$ of a dataflow $W$, denoted as $\mathcal{S}(W, t)$, is a pair $(\mathcal{S}, \preccurlyeq_\mathcal{S})$, where $\mathcal{S}$ is a set of tuples and $\preccurlyeq_\mathcal{S}$ is a partial order on $\mathcal{S}$, defined as follows: 
\begin{enumerate}
    \item The source tuple is in $\mathcal{S}$.
    \item For each tuple $s$ in $\mathcal{S}$, if an operator processes the tuple $s$ and produces zero or more output tuples $\{s'_1, \ldots, s'_n\}$, all the produced tuples are also in $\mathcal{S}$. 
    For each tuple $s'_i$, we have the order $s \prec s'_i$ in $\preccurlyeq_\mathcal{S}$.
\end{enumerate}
\end{definition}

For instance, in Figure~\ref{fig:scope-example}, a source tuple $t$ is ingested into the dataflow from the source operator $A$ and processed by operators $C$, $D$, $E$, $F$, and $H$. The scope of $t$  includes the tuples on the highlighted edges and their partial order defined as their edges on the DAG. 

\begin{figure}[htbp]
	\includegraphics[width=2.0in]{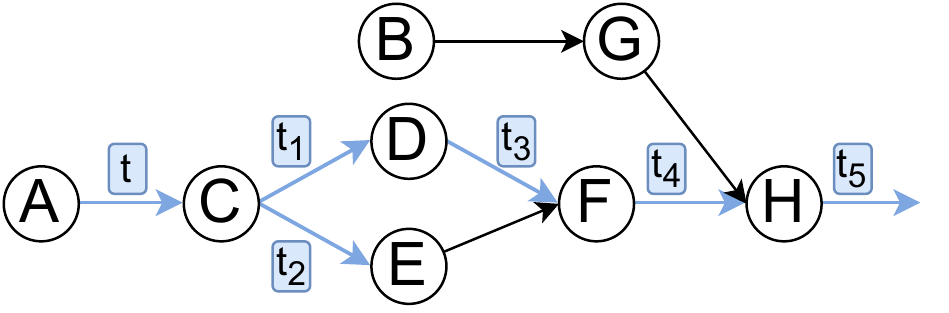}
	\caption{\label{fig:scope-example}
		\textbf{Scope of a source tuple in a dataflow.}
	}
\end{figure}

\begin{definition}[Data operation]
The {\em data operation} of a tuple $s$ is the processing of $s$ by its receiving operator $o$, denoted as $\phi(s, o)$.
\end{definition}

\begin{definition}[Data transaction]
\label{def:data-transaction}
For a dataflow $W$ and a source tuple $t$ in $W$, let $(\mathcal{S}, \preccurlyeq_\mathcal{S})$ be the scope of $t$. 
The {\em data transaction} of $t$ is a pair $(\Phi, \preccurlyeq_{\Phi})$, where $\Phi$ is the set of data operations of the tuples in $\mathcal{S}$, and $\preccurlyeq_{\Phi}$ is a partial order on $\Phi$. 
For two data operations $\phi(t_i, o_i)$ and $\phi(t_j, o_j)$ in $\Phi$, we have $\phi(t_i, o_i) \prec \phi(t_j, o_j)$ in $\preccurlyeq_{\Phi}$ if and only if $t_i \prec t_j$ is in $\preccurlyeq_\mathcal{S}$.
\end{definition}

For instance, in Figure~\ref{fig:fd-reconfig-multi-op-dcm}, tuple $t$ has the following data transaction $T_1$: 
$$T_1: [\phi(t, FC), \phi(t, FM), \phi(t, MC)].$$ 
In the data transaction, ``$\phi(t, FC)$'' is a data operation representing the processing of this tuple $t$ by the $FC$ operator.

\begin{definition}[Function-update transaction]
\label{def:function-update-transaction}
The {\em function-update transaction} of a reconfiguration $\mathcal{R} = \{ (o_1,\mu(o_1)),\ldots, (o_n,\mu(o_n)) \}$ on a dataflow $W$ is the set $\{\mu(o_1),\ldots,\mu(o_n)\}$, where each $\mu(o_i)$ is a function-update operation in $\mathcal{R}$.
\end{definition}

For instance, the reconfiguration in Figure~\ref{fig:fd-reconfig-multi-op-dcm} has the following function-update transaction $T_2$: 
$$T_2: \{\mu(FM), \mu(MC)\}.$$
In the function-update transaction, ``$\mu(FM)$'' is a function-update operation representing that the operator $FM$ switches to the new configuration.
Note that the order of different operations in a function-update transaction does not matter because they update different operators and are independent of each other. 

\begin{definition}[Conflicting operations]
\label{def:conflicting-operations}
A data operation $\phi(t, o)$ and a function-update operation $\mu(o')$ are said to be {\em conflicting} if $o = o'$, i.e., they are on the same operator.  They are said to be {\em not conflicting} if $o \ne o'$.
\end{definition}

For instance, in Figure~\ref{fig:fd-reconfig-multi-op-dcm}, operations $\phi(t, FM)$ and $\mu(FM)$ are conflicting because they are on the same operator. Operations $\phi(t, FC)$ and $\mu(FM)$ are not conflicting as they are on different operators.

\begin{definition}[Schedule]
\label{def:transaction-schedule}
A {\em schedule} of a set of transactions $T_1,\ldots,T_k$ is the set of all the operations in those transactions with a {\em partial order}. The schedule is called {\em serial} if for each pair of transactions $T_i$ and $T_j$,  $T_i$'s operations in the schedule are either all before those in $T_j$ or all after those in $T_j$.
\end{definition}

In this paper we only consider schedules that include one function-update transaction and many data transactions.

\begin{definition}[Conflict-equivalence]
\label{def:conflict-equivalence}
Two schedules $S_1$ and $S_2$ of the same set of transactions are said to be {\em conflict-equivalent} if $\forall o_i, o_j \in S_1$, if $o_i$ and $o_j$ are conflicting, and $o_i$ is before $o_j$ in $S_1$, then $o_i$ is also before $o_j$ in $S_2$.
\end{definition}

\begin{definition}[Conflict-serializable]
\label{def:conflict-serializable}
A schedule is said to be {\em conflict-serializable} if it is conflict-equivalent to a serial schedule of the same set of transactions. 
\end{definition}

In the rest of the paper, when a partial order of a data transaction or a schedule defines a total order, for simplicity, we just show the transaction or the schedule as a sequence.  We use the running example in Figure~\ref{fig:fd-reconfig-multi-op} to explain these concepts. 

\begin{itemize}
    \item $S_1$ is a schedule of the two transactions $T_1$ and $T_2$: $$S_1: [\phi(t, FC),  \mu(FM), \phi(t, FM), \mu(MC), \phi(t, MC)].$$
    \item $S_2$ is a serial schedule of the two transactions: $$S_2: [\mu(FM), \mu(MC), \phi(t, FC), \phi(t, FM), \phi(t, MC)].$$
    In particular, all $T_2$'s operations in this schedule are before those in  $T_1$.
    \item $S_1$ and $S_2$ are conflict-equivalent. For example, for the conflicting pair $\mu(FM)$ and $\phi(t, FM)$, the former is before the latter in both schedules.
    \item $S_1$ is conflict-serializable because it is conflict-equivalent to the serial schedule $S_2$.
    \item $S_3$ is not a conflict-serializable schedule: $$S_3: [\phi(t, FC),  \phi(t, FM), \mu(FM), \mu(MC), \phi(t, MC)].$$ We can show that $S_3$ is not conflict-equivalent to any serial schedule. Intuitively, it has two pairs of conflicting operations, namely [$\phi(t, FM)$, $\mu(FM)$] and [$\mu(MC)$, $\phi(t, MC)$], and their corresponding transaction orders are different. 
\end{itemize}

$S_3$ is the ``bad'' schedule generated by the naive FCM scheduler in Section~\ref{subsec:naive-fcm-scheduler}, in which tuple $t$ is processed using the old configuration of $FM$ and the new configuration of $MC$.  Schedule $S_1$ is a ``good'' schedule since $t$ is processed entirely using the new configurations of both operators $FM$ and $MC$ and the aforementioned schema-mismatch issue does not happen.

\boldstart{Consistency of epoch-based schedulers.}  Consider the example in Figure~\ref{fig:fd-reconfig-multi-op}. The aforementioned schedule $S1$ in Section~\ref{subsec:reconfig-consistency} is produced by the EBR epoch-based scheduler, where the epoch marker is propagated before tuple $t$. We show that the EBR approach can always produce a conflict-serializable schedule in Lemma~\ref{lemma:epoch-multi}. 
We also show that in general, an epoch-based scheduler always produces conflict-serializable schedules in Lemma~\ref{lemma:epoch-general-multi}.

\begin{lemma}
\label{lemma:epoch-multi}
Every schedule produced by the EBR epoch-based scheduler is conflict-serializable.
\end{lemma}

\begin{proof}
Let $S$ be a produced schedule. We construct a serial schedule $S'$ using the following steps. Consider the function-update transaction $U$ and each data transaction $T$ for a tuple $t$. Since the epoch marker serve as a barrier dividing the input stream into two epochs, $t$ can  be in only one of the following two cases: 

\begin{itemize}
  \item $t$ is before the epoch marker. For each conflict in $S$ between a data operation $\phi$ in $T$ and a function-update operation $\mu$ in $U$, $\phi$ is before $\mu$. We place $T$ before $U$ in $S'$. Thus, in $S'$, $\phi$ is also before $\mu$.
  \item $t$ is after the epoch marker. Similarly, we place $T$ after $U$ in $S'$. Each conflict order in $S$ remains the same in $S'$.
\end{itemize}

For those data transactions before $U$, we order them in $S'$ following the order of their first data operations in $S$. For data transactions after $U$, we order them in $S'$ following the order of their first data operations in $S$. Notice that there are no conflicts between two data transactions.

The schedule $S$ is conflict-equivalent to the constructed serial schedule $S'$ because all conflicting pairs in $S$ have the same order in $S'$. Therefore, $S$ is conflict-serializable.
\end{proof}

\begin{lemma}
\label{lemma:epoch-general-multi}
Every schedule produced by a general epoch-based scheduler is conflict-serializable.
\end{lemma}

\begin{proof}
Consider a function-update transaction $U$ and each data transaction $T$ for a source tuple $t$. For a general epoch-based scheduler, all the tuples in the scope of $t$ are in one epoch $E_t$ and $U$ is scheduled between two epochs $E_{i}$ and $E_{i+1}$. We can compare $E_t$ with the boundary between $E_{i}$ and $E_{i+1}$ to determine the position of $T$ in a serial schedule.
We can prove this claim using steps similar to those in the proof of Lemma~\ref{lemma:epoch-multi}.
\end{proof}

\ifshort

\else
\begin{proof}
Let $S$ be a produced schedule. We construct a serial schedule $S'$ using the following steps. Consider the function-update transaction $U$ and each data transaction $T$ for a tuple $t$. Since the epoch marker serve as a barrier dividing the input stream into two epochs, $t$ can  be in only one of the following two cases: 

\begin{itemize}
  \item $t$ is before the epoch marker. For each conflict in $S$ between a data operation $\phi$ in $T$ and a function-update operation $\mu$ in $U$, $\phi$ is before $\mu$. We place $T$ before $U$ in $S'$. Thus, in $S'$, $\phi$ is also before $\mu$.
  \item $t$ is after the epoch marker. Similarly, we place $T$ after $U$ in $S'$. Each conflict order in $S$ remains the same in $S'$.
\end{itemize}

For those data transactions before $U$, we order them in $S'$ following the order of their first data operations in $S$. For data transactions after $U$, we order them in $S'$ following the order of their first data operations in $S$. Notice that there are no conflicts between two data transactions.

The schedule $S$ is conflict-equivalent to the constructed serial schedule $S'$ because all conflicting pairs in $S$ have the same order in $S'$. Therefore, $S$ is conflict-serializable.
\end{proof}

\fi

\section{Dataflows with one-to-one Operators Only}
\label{sec:hybrid-single-tuple}

In this section, we consider the case where a dataflow contains one-to-one operators only. We propose a scheduler called \sysname, which uses FCM's to achieve low reconfiguration delay and still guarantees conflict-serializability of produced schedules. 

\begin{definition}[One-to-one operator]
\label{def:one-to-one-op}
An operator is called {\em one-to-one} if its processing function emits at most one (tuple, receiving operator) pair for each input tuple.
\end{definition}
This type includes operators such as projection, filter, map function, equi-join on key attributes, and union.

\begin{definition}[One-to-many operator]
\label{def:one-to-many-op}
An operator is called {\em one-to-many} if its processing function can emit more than one output (tuple, receiving operator) pair for an input tuple.
\end{definition}
This type includes operators such as join on non-key attributes and flatten function.
In the rest of this section, we consider dataflows where all operators in the dataflow are one-to-one.

\subsection{Conflict-Serializable Schedules Produced by the Naive FCM-based Scheduler}
\label{subsec:dcm-multi-op}

Section~\ref{subsec:naive-fcm-scheduler} shows an example dataflow and a reconfiguration where the naive FCM-based scheduler produces a non-conflict-serializable schedule. Next we use an example to show that the naive scheduler can still guarantee conflict-serializability for some types of dataflows and reconfigurations.

\begin{example}
\label{example:dcm-multi-independent}

Suppose we want to use the naive FCM-based scheduler to handle a reconfiguration of the two operators $C$ and $D$ as shown in Figure~\ref{fig:multi-op-dcm}. Operator $X$ is a one-to-one operator that splits the outputs tuples to operators $C$ and $D$.  In this case, we have a data transaction $T_3=[\phi(t_1, X),\phi(t_1, C])$, another data transaction $T_4=[\phi(t_2, X),\phi(t_2, D])$, and a function-update transaction $U=$ $[\mu(C),$ $\mu(D)]$.  The controller sends two separate FCM's to $C$ and $D$. Consider a possible schedule with $T_3$, $T_4$, and $U$:
$$S_4: [\phi(t_1, X), \mu(C), \phi(t_1, C), \phi(t_2, X), \mu(D), \phi(t_2, D)].$$
Schedule $S_4$ is conflict-serializable because it is conflict-equivalent to the serial schedule $[U, T_3, T_4]$. Interestingly, we can show that all schedules produced by the naive FCM-based scheduler in Figure~\ref{fig:multi-op-dcm} are conflict-serializable.

\end{example}

\begin{figure}[htbp]
	\includegraphics[width=2.0in]{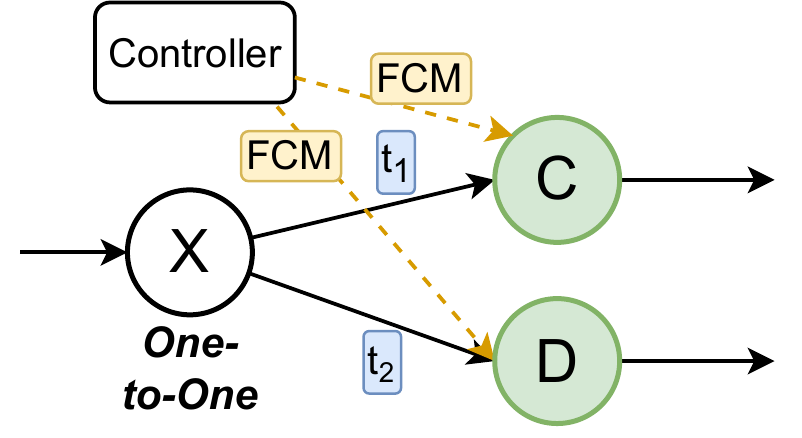}
	\caption{\label{fig:multi-op-dcm}
		\textbf{An example dataflow with a reconfiguration on operators $C$ and $D$. The naive FCM-based scheduler always produces a conflict-serializable schedule. }
	}
\end{figure}

\ifshort 

\else
\begin{proof}
Let $S$ be a produced schedule. We construct a serial schedule $S'$ using the following steps.
Consider the function-update transaction $U$ and each data transaction $T$ for a tuple $t$. In a dataflow with one-to-one operators only, tuple $t$ processed by operator $X$ can be either dropped by $X$, or output to either operator $C$ or operator $D$. Therefore, there are two possible cases:
\begin{itemize}
  \item $T$ has no conflict with $U$. We place $T$ before $U$ in $S'$. 
  \item $T$ has a single conflicting operation with $U$ on either $C$ or $D$, but not both. If the data operation in the conflict is before the corresponding function-update operation, we place $T$ before $U$ in $S'$. Otherwise, we place $T$ after $U$ in $S'$.
\end{itemize}

For those data transactions before $U$, we order them in $S'$ following the order of their first data operations in $S$. For data transactions after $U$, we order them in $S'$ following the order of their first data operations in $S$. Notice that there are no conflicts between two data transactions.

The schedule $S$ is conflict-equivalent to the constructed serial schedule $S'$ because all conflicting pairs in $S$ have the same order in $S'$. Therefore, $S$ is conflict-serializable.
\end{proof}
\fi

One might wonder why the two examples in Figure~\ref{fig:fd-reconfig-multi-op-dcm} and Figure~\ref{fig:multi-op-dcm} are different in the conflict-serializability of the produced schedules.  The main reason is that in Figure~\ref{fig:fd-reconfig-multi-op-dcm}, a tuple can be processed by operators $FM$ and $MC$, and both of them are in the reconfiguration. But there is no synchronization between the data operations and the function-update operations, causing the non-conflict-serializability. While in Figure~\ref{fig:multi-op-dcm}, a tuple is processed by only one of the two paths through either $C$ or $D$. On each path, there is a single operator in the reconfiguration, thus the data operations and the function-update operations are always synchronized.

\label{sec:hybrid-multi-op}
Next, we introduce a concept called ``minimal covering sub-DAG,'' which is used to represent the synchronization components. We then describe the \sysname scheduler using this concept, and prove that this scheduler can always produce a conflict-serializable schedule.

\subsection{Minimal Covering Sub-DAG (MCS)}

\begin{definition}[Minimal covering sub-DAG]
\label{def:minimal-covering-sub-dag}
Given a DAG $G = (V, E)$, and a set of vertices $M \subseteq V$, a minimal covering sub-DAG $G' = (V', E')$ is defined as follows:

\begin{enumerate}
  \item $M \subseteq V'$;
  \item $\forall A, B \in M$, if there is a path from $A$ to $B$, then all the vertices and edges on the path are in $V'$ and $E'$, respectively;
  \item $G'$ is minimal, i.e., no proper sub-DAG of $G'$ can satisfy the above two conditions.
\end{enumerate}

\end{definition}

\begin{lemma}
\label{lemma:mcs-unique}
There is a unique MCS given a DAG and a set of vertices. 
\end{lemma}
\begin{proof}
Suppose $G'_1$ and $G'_2$ are two distinct MCS's of a DAG and a set of vertices $M$. Consider the sub-DAG $G'_3$ that is the ``intersection'' of $G'_1$ and $G'_2$, i.e., the vertices of $G'_3$ are the intersection of the two sets of vertices in $G'_1$ and $G'_2$, the edges of $G'_3$ are the intersection of the two sets of edges in $G'_1$ and $G'_2$. Since $M$ is a subset of the sets of vertices in $G'_1$ and $G'_2$, $M$ is also a subset of the vertices in $G'_3$. Thus $G'_3$ satisfies property (1).  $\forall A, B \in M$, if there is a path $p$ from $A$ to $B$, $p$ is also in both $G'_1$ and $G'_2$, so $p$ is also in the $G'_3$. Thus $G'_3$ also satisfies property (2). 
Therefore, $G'_3$ is also an MCS, which contradicts the minimality property (3) of $G'_1$ and $G'_2$. 
\end{proof}

Algorithm~\ref{alg:find-mcs} shows an algorithm for finding the minimal covering sub-DAG (MCS) given a DAG $G$ and a set of vertices $M$. 
In lines~\ref{alg:find-mcs-forward-begin}-~\ref{alg:find-mcs-forward-end}, we iterate through the DAG in a topological order. For each vertex $v$, we mark $v$ in ``red'' if $v$ is in $M$ or any parent vertex of $v$ is marked in ``red.'' After this iteration, a vertex marked in ``red'' is either 1) in $M$, or 2) a descendant of a vertex in $M$. 
Next in lines~\ref{alg:find-mcs-backward-begin}-~\ref{alg:find-mcs-backward-end}, we iterate through the DAG in a reverse topological order. For each vertex $v$, we mark $v$ in ``blue'' if $v$ is in $M$ or any child of $v$ is marked in ``blue.'' After this iteration, a vertex marked in ``blue'' is either 1) in $M$, or 2) is an ancestor of a vertex in $M$.
Finally, in lines~\ref{alg:find-mcs-vertices-begin}-~\ref{alg:find-mcs-vertices-end},  we add all the vertices marked in both ``red'' and ``blue'' to the MCS because these vertices are either 1) in $M$, or 2) on a path between two operators in $M$. Then we add all the edges connecting these vertices to the MCS. 

The time complexity of this algorithm is $O(V + E)$, specifically: 
\begin{itemize}
    \item The topological ordering in line~\ref{alg:find-mcs-topo} takes $O(V + E)$~\cite{books/daglib/0032835};
    \item The loop of marking ``red'' in lines~\ref{alg:find-mcs-forward-begin}-~\ref{alg:find-mcs-forward-end} takes $O(V + E)$ because we iterate through each vertex once and look at every edge once;
    \item Similarly, the loop of marking ``blue'' in lines~\ref{alg:find-mcs-backward-begin}-~\ref{alg:find-mcs-backward-end} also takes $O(V + E)$;
    \item The loop in lines lines~\ref{alg:find-mcs-vertices-begin}-~\ref{alg:find-mcs-vertices-end} takes $O(V)$ and the loop in lines lines~\ref{alg:find-mcs-edges-begin}-~\ref{alg:find-mcs-edges-end} takes $O(E)$.
\end{itemize}

\begin{algorithm}
\caption{Find Minimal Covering SubDAG}\label{alg:find-mcs}
\algrenewcommand\algorithmicrequire{\textbf{Input:}}
\algrenewcommand\algorithmicensure{\textbf{Output:}}
\begin{algorithmic}[1]
\Require{dataflow DAG $G=(V,E)$}
\Require{$M = \{m_1,\ldots,m_n\}$}
    \State $D \leftarrow \varnothing$
    \For{each $v \in V$}
        \State $D[v] \leftarrow \varnothing$
    \EndFor
    \State let $v_1, \ldots, v_{|V|}$ be a topological ordering of $G$ \label{alg:find-mcs-topo}
    \For{each $v \leftarrow v_1, \ldots, v_{|V|}$} \label{alg:find-mcs-forward-begin}
        \If{$v \in M$}
            \State add ``red'' to $D[v]$
        \EndIf
        \For{each incoming edge $e$ of $v$}
            \If{``red'' $\in D[e.from]$}
                \State add ``red'' to $D[v]$
            \EndIf
        \EndFor 
    \EndFor \label{alg:find-mcs-forward-end}
    \For{each $v \leftarrow v_{|V|}, \ldots, v_1$} \label{alg:find-mcs-backward-begin}
        \If{$v \in M$}
            \State add ``blue'' to $D[v]$
        \EndIf
        \For{each outgoing edge $e$ of $v$}
            \If{``blue'' $\in D[e.to]$}
                \State add ``blue'' to $D[v]$
            \EndIf
        \EndFor 
    \EndFor \label{alg:find-mcs-backward-end}
    \State $V'=\{\}$ 
    \For{each $v \in V$} \label{alg:find-mcs-vertices-begin}
        \If{$D[v] =$ \{``red'', ``blue''\}}
            \State add $v$ to $V'$
        \EndIf
    \EndFor \label{alg:find-mcs-vertices-end}
    \State $E'=\{\}$ 
    \For{each $e \in E$} \label{alg:find-mcs-edges-begin}
        \If{$e.from \in V' \land e.to \in V'$}
            \State add $e$ to $E'$
        \EndIf
    \EndFor \label{alg:find-mcs-edges-end}
    \State \textbf{return} $(V', E')$
\end{algorithmic}
\end{algorithm}

Figure~\ref{fig:hybrid-min-subdag} shows the minimal covering sub-DAG for the dataflow graph in Figure~\ref{fig:scope-example} and the set of operators $\{C,F,G\}$ in the reconfiguration. The sub-DAG is: $V' = \{C,D,E,F,G\}$ and $E'=\{C{\rightarrow}D,C{\rightarrow}E    ,D{\rightarrow}F,E{\rightarrow}F\}$.  In general, we can show that there is a unique MCS given a DAG and a set of vertices, and we can compute the MCS using an algorithm with an $O(V + E)$ time complexity.

\begin{figure}[htbp]
	\includegraphics[width=2.5in]{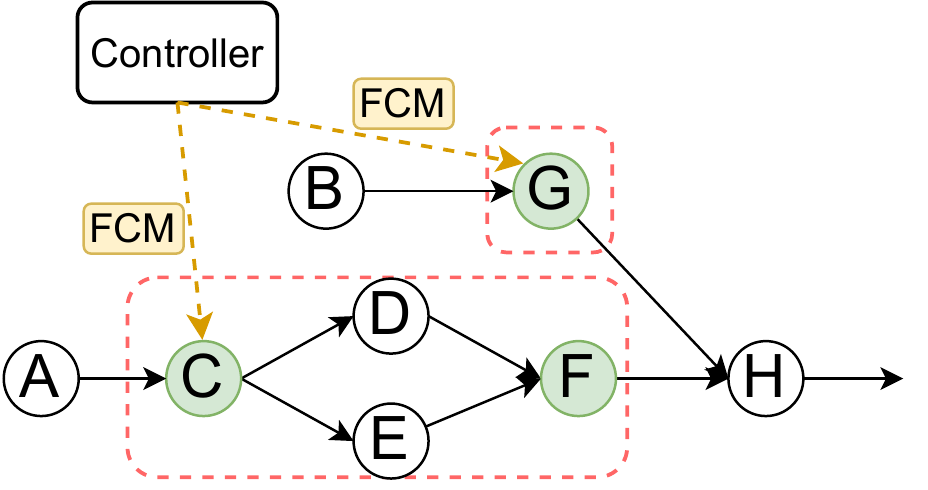}
	\caption{\label{fig:hybrid-min-subdag}
	Two components of the minimal covering sub-DAG used in the \sysname scheduler are highlighted in red.
	}
\end{figure}

\subsection{The \sysname Scheduler}
The \sysname scheduler uses components of the MCS to schedule the reconfiguration. A {\em component} is a maximal sub-DAG of the MCS where every pair of vertices in the component are connected by a path, ignoring the direction of edges. For example, the sub-DAG in Figure~\ref{fig:hybrid-min-subdag} has two components, each marked in a red box. The components of the MCS can be also computed using an algorithm~\cite{books/daglib/0017733} with an $O(V+E)$ time complexity.

The \sysname scheduler is formally described in Algorithm~\ref{alg:hybrid-reconfig}. We first construct the minimal covering sub-DAG from the original dataflow DAG and operators in the reconfiguration (lines~\ref{alg:hybrid:line:find-mcs-start} and~\ref{alg:hybrid:line:find-mcs-end}). 
We compute the components within the MCS (line~\ref{alg:hybrid:line:find-components}).
For each component in the MCS, the controller sends an FCM to the ``head'' operators, i.e., those with no input edges in the component. The head operators then start propagating an epoch marker within the component (lines~\ref{alg:hybrid:line:epoch-start} to~\ref{alg:hybrid:line:epoch-end}).  Specifically, when an operator receives an epoch marker, it performs marker alignment on the input edges in its component. An operator sends an epoch marker only to its downstream operators in its component.

\begin{algorithm}
\caption{The \sysname Scheduler (for dataflows with one-to-one operators only)}\label{alg:hybrid-reconfig}
\algrenewcommand\algorithmicrequire{\textbf{Input:}}
\algrenewcommand\algorithmicensure{\textbf{Output:}}
\begin{algorithmic}[1]
\Require{$G=(V,E)$}
\Require{$\mathcal{R} = \{ (o_1,U_1),\ldots, (o_n,U_n) \}$}
    \State $M \leftarrow \{o_1, \ldots, o_n\}$ \label{alg:hybrid:line:find-mcs-start}
    \State $G' \leftarrow findMCS(G, M)$ \label{alg:hybrid:line:find-mcs-end}
    \State $\mathcal{C}_1, \ldots, \mathcal{C}_p \leftarrow findComponents(G')$ \label{alg:hybrid:line:find-components}
    
    \For {each $\mathcal{C} \leftarrow \mathcal{C}_1, \ldots, \mathcal{C}_p$} \label{alg:hybrid:line:epoch-start}
        \State send an FCM to the each head operator in $\mathcal{C}$
        \State start propagating an epoch marker within $\mathcal{C}$
    \EndFor \label{alg:hybrid:line:epoch-end}
\end{algorithmic}
\end{algorithm}

As an example, in Figure~\ref{fig:hybrid-min-subdag}, the controller sends an FCM to operator $C$, which is the only head operator of the first component. The controller also sends an FCM to operator $G$, which is the only head operator of the second component.
When $C$ receives the FCM, it applies the new configuration and starts propagating an epoch marker to operators $D$ and $E$. These operators then forward the marker to operator $F$. When $F$ receives the marker from both $D$ and $E$, it applies the new configuration and stops the marker propagation. When operator $G$ receives the marker, it applies the new configuration and does not send out an epoch marker. 

Next, we show that the \sysname scheduler can always produce a conflict-serializable schedule.

\begin{lemma}
\label{lemma:mcs-component}
Consider a dataflow graph $G$ with one-to-one operators only, with a reconfiguration $\mathcal{R}$, and the MCS $G'$ generated by Algorithm~\ref{alg:hybrid-reconfig}.
Each component of $G'$ contains at least one reconfiguration operator.
\end{lemma}

\begin{proof}
By the construction of Algorithm~\ref{alg:hybrid-reconfig}, the set of reconfiguration operators $M$ in $\mathcal{R}$ are used to construct the MCS $G'$.
Suppose all the vertices in a component $\mathcal{C}$ of $G'$ are not in $M$. We construct a new DAG $G''$ by removing all the vertices and edges in $C$ from $G'$. 
Using similar steps as in Lemma~\ref{lemma:mcs-unique}, we can show that $G''$ is still a minimal covering sub-DAG of $G$ and $M$. This result contradicts the minimality property of $G'$.
\end{proof}

\begin{lemma}
\label{lemma:hybrid-component}
In dataflows with one-to-one operators only, consider a dataflow graph $G$ with a reconfiguration $\mathcal{R}$, and the MCS $G'$ generated by Algorithm~\ref{alg:hybrid-reconfig}. For each source tuple, the operators in its data transaction overlap with at most one component of $G'$.
\end{lemma}

\begin{proof}
Suppose the operators processing a tuple $t$ overlap with two components $\mathcal{C}_1$ and $\mathcal{C}_2$ in the MCS. Based on Lemma~\ref{lemma:mcs-component}, there is an operator $A$ in $\mathcal{C}_1$ and another operator $B$ in $\mathcal{C}_2$, where $A, B \in M$.
In dataflows with one-to-one operators only, tuple $t$ goes through a {\em chain} of operators and there must be a path between $A$ and $B$ in $G$. By Definition~\ref{def:minimal-covering-sub-dag}, the path must also be in $G'$. By the definition of components, $A$ and $B$ must be in the same component of $G'$, which contradicts the assumption.
\end{proof}

\begin{theorem}
\label{theorem:hybrid-method}
In dataflows with one-to-one operators only, the \sysname scheduler in Algorithm~\ref{alg:hybrid-reconfig} always produces a conflict-serializable schedule.
\end{theorem}

\begin{proof}
Let $S$ be a produced schedule. We can construct a serial schedule $S'$ using the following steps. Consider the function-update transaction $U$ and each data transaction $T$ for a tuple $t$. Based on Lemma~\ref{lemma:hybrid-component}, $T$ can be in only one of the following two cases. (1) Operators in $T$ do not overlap with any component in the MCS. In this case, $T$ does not have any conflict with $U$. We can place $T$ before $U$ in $S'$. (2) Operators in $T$ overlap with one component in the MCS. In this case, we place $T$ in $S'$ 
by comparing the position of a tuple and the epoch marker of this component.  In both cases, for those data transactions before $U$, we order them in $S'$ following the order of their first data operations in $S$. For data transactions after $U$, we order them in $S'$ following the order of their first data operations in $S$. Notice that there are no conflicts between two data transactions.
The schedule $S$ is conflict-equivalent to the constructed serial schedule $S'$ because all conflicting pairs in $S$ have the same order in $S'$. Therefore, $S$ is conflict-serializable.
\end{proof}

The reconfiguration delay of the Fries scheduler is decided by the size of each MCS component, which is the number of edges in the component. Compared to the EBR scheduler, the FCMs sent to the head of each MCS component are not blocked by the processing of data by the upstream operators.
Within each MCS component, the Fries scheduler still relies on epoch markers. In the extreme case where the MCS covers the entire dataflow graph, the Fries scheduler essentially becomes the epoch-based scheduler, where the FCMs are sent to all source operators and the epoch markers need to be propagated through the entire dataflow.

\section{Dataflows with One-to-Many Operators}
\label{sec:complex-extension}
In this section we consider dataflows with one-to-many operators.

\subsection{Challenges}
\label{subsec:one-to-many}

Figure~\ref{fig:extension-example-complex} shows a part of a dataflow with a one-to-many {\sf Join} operator, which joins each input tuple with the {\sf Merchants} table. 
When a tuple contains purchases from multiple merchants, {\sf Join} generates multiple output tuples. For instance, the tuple $t_1$ joins with three merchants and produces the tuples $t_2$, $t_3$, and $t_4$. 
The {\sf Split} operator splits the stream based on merchant information and sends different tuples to the two merchant fraud-detector operators $FMX$ and $FMY$. The prediction results are combined by a {\sf Union} operator.

\begin{figure}[htbp]
	\includegraphics[width=\linewidth]{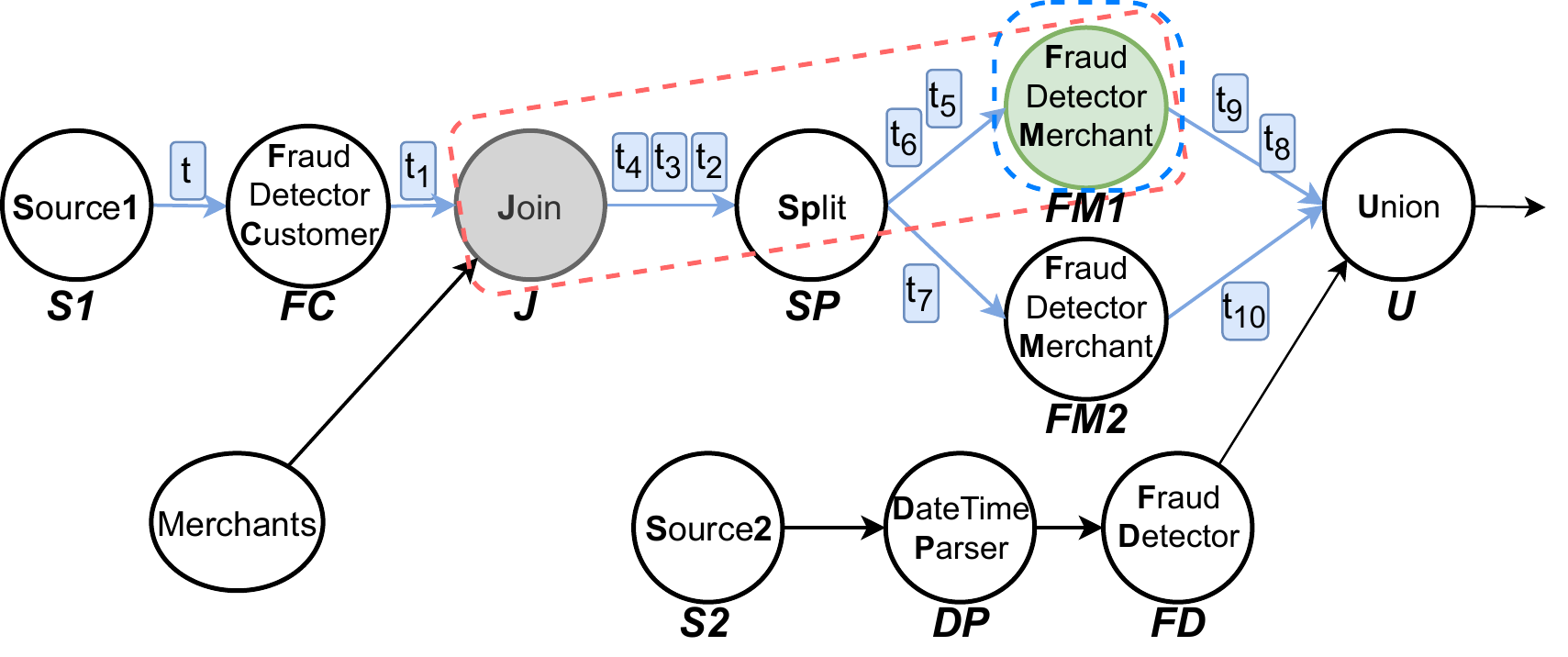}
	\caption{\label{fig:extension-example-complex}
		\textbf{Reconfiguration of operator $FM1$ in a dataflow with a one-to-many {\sf Join} operator. An incorrect MCS generated by Algorithm~\ref{alg:hybrid-reconfig} is highlighted in blue. The correct MCS generated by Algorithm~\ref{alg:hybrid-reconfig-extended} is highlighted in red.} 
	}
\end{figure}

Based on Definition~\ref{def:data-transaction}, the source tuple $t$ has the following data transaction $T_5$. 
\begin{multline*}
\Phi \; in \; T_5: \{\phi(FC, t), \phi(J, t_1), \phi(SP, t_2), \phi(SP, t_3),  \phi(SP, t_4), \\ \phi(FMX, t_5), \phi(FMX, t_6), \phi(FMY, t_7),  \phi(U_1, t_8), \phi(U_1, t_9), \phi(U_1, t_{10}) \}. 
\end{multline*}

We use an example to show that when reconfiguring a dataflow with one-to-many operators, a naive adoption of the \sysname scheduler in Algorithm~\ref{alg:hybrid-reconfig} can produce a non-conflict-serializable schedule.
Consider a reconfiguration of operator $FMX$ in Figure~\ref{fig:extension-example-complex}. Algorithm~\ref{alg:hybrid-reconfig} adds the only reconfiguration operator $FMX$ to the set $M$ and computes the MCS with one component, which contains the operator $FMX$ and no other edges. Algorithm~\ref{alg:hybrid-reconfig} ignores the {\sf Join} operator because it is not in the reconfiguration.
The method sends an FCM to $FMX$. This operator does not propagate the FCM to its downstream operators because it is the only operator in the MCS component.
Suppose the FCM sent to operator $FMX$ arrives {\em after} the tuple $t_5$ and before the tuple $t_6$ in the same transaction. Then this scheduler produces the following schedule with a total order of the data operations and the function-update operations:
\begin{multline*}
S_5: [\phi(FC, t), \phi(J, t_1), \phi(SP, t_2), \phi(SP, t_3),  \phi(SP, t_4), \boldsymbol{\phi(FMX, t_5)},  \\ \boldsymbol{\mu(FD_1)}, \boldsymbol{\phi(FMX, t_6)}, \phi(FMY, t_7),  \phi(U, t_8), \phi(U, t_9), \phi(U, t_{10}) ].
\end{multline*}

We can show that the schedule $S_5$ is not conflict-serializable. Intuitively, as indicated in the operations in bold, tuple $t_5$ is processed by $FMX$ with the old configuration, and tuple $t_6$ in the same transaction is processed by $FD_2$ with the new configuration.

\subsection{Extending the \sysname scheduler}
\label{subsec:one-to-many-extension}

We extend the \sysname scheduler Algorithm~\ref{alg:hybrid-reconfig} to produce a conflict-serializable schedule for a dataflow with one-to-many operators and a function-update transaction.
Intuitively, for a one-to-many operator, each of its descendant operators could receive multiple input tuples that belong to the same data transaction. 
In Figure~\ref{fig:extension-example-complex}, operator $SP$ receives three tuples ($t_2$, $t_3$, and $t_4)$, and operator $FMX$ receives two tuples ($t_5$ and $t_6$) in the same data transaction.

Consider a reconfiguration that includes the operator $FD_1$. The function-update operation $\mu(FD_1)$ can be conflicting with the data operations of tuples $t_5$ and $t_6$ (in the same data transaction) in the same operator.
To guarantee a conflict-serializable schedule, these two data operations must synchronize with $\mu(FMX)$ to ensure that both data operations are either before $\mu(FMX)$ or after $\mu(FMX)$. In other words, $\mu(FMX)$ cannot be scheduled in the middle of these two data operations. Notice that the {\sf Join} operator is the earliest ancestor one-to-many operator of the reconfiguration operator $FMX$. 
If an FCM is sent to an operator $O$ after the {\sf Join} operator, since the operator $O$ could possibly generate multiple data operations for the same data transaction, the FCM can be injected in the middle of these data operations, causing the schedule to be not conflict-serializable.
Based on these observations, to guarantee the conflict-serializability, we can start the synchronization from the {\sf Join} operator using an epoch marker. 
Recall that the \sysname scheduler starts the epoch marker propagation from the head operators of a component in the MCS. The MCS is constructed using a set of operators $M$, which includes the reconfiguration operator $FMX$.
To make sure the {\sf Join} operator is treated as a head operator in a component, we add the operator to $M$ before computing the MCS.

\begin{algorithm}
\caption{The \sysname Scheduler (for general dataflows with one-to-many operators)}\label{alg:hybrid-reconfig-extended}
\algrenewcommand\algorithmicrequire{\textbf{Input:}}
\algrenewcommand\algorithmicensure{\textbf{Output:}}
\begin{algorithmic}[1]
\Require{A dataflow $G=(V,E)$}
\Require{A reconfiguration $\mathcal{R} = \{ (o_1,U_1),\ldots, (o_n,U_n) \}$}
    \State $M = \{o_1, \ldots, o_n\}$ \label{alg:hybrid:line:mcs-init}
    \For {each reconfiguration operator $o_i$ in $\{o_1, \ldots, o_n\}$} \tikzmark{pa} \label{alg:hybrid:line:one-to-many-start}
        \State $\mathcal{A} \leftarrow$ set of ancestor one-to-many operators of $o_i$
        \State $\mathcal{E} \leftarrow computeEarliestAncestors(\mathcal{A}$)
        \State $M \leftarrow M \cup \mathcal{E}$
        \label{alg:hybrid:line:one-to-many-body}
    \EndFor \tikzmark{pb} 
    
    \State \ldots same as Algorithm~\ref{alg:hybrid-reconfig} line \ref{alg:hybrid:line:find-mcs-end}-\ref{alg:hybrid:line:epoch-end}

\end{algorithmic}
\begin{tikzpicture}[remember picture,overlay]
\draw[black,rounded corners]
  ([shift={(-205pt,2ex)}]pic cs:pa) 
    rectangle 
  ([shift={(215pt,-0.65ex)}]pic cs:pb);
\end{tikzpicture}
\end{algorithm}

Algorithm~\ref{alg:hybrid-reconfig-extended} shows the extended \sysname scheduler, with the  part in the box showing the differences compared to the original \sysname scheduler in Algorithm~\ref{alg:hybrid-reconfig}. 
When constructing the MCS, apart from adding the operators in the reconfiguration to $M$ (line~\ref{alg:hybrid:line:mcs-init}), we also add to $M$ all the earliest one-to-many ancestor operators of each reconfiguration operator $o_i$ (lines~\ref{alg:hybrid:line:one-to-many-start} to~\ref{alg:hybrid:line:one-to-many-body}). This step is done by first finding the set of ancestor one-to-many operators of $o_i$, denoted as $\mathcal{A}$, then finding the earliest ones in $\mathcal{A}$.
Notice that a reconfiguration operator could have more than one earliest ancestor one-to-many operator. For example, in Figure~\ref{fig:extension-example-complex}, suppose the operators $FMX$ and $FMY$ are the only one-to-many operators in the dataflow. Then the reconfiguration operator $U$ has both $FMX$ and $FMY$ as its earliest ancestor one-to-many operators according to the partial order of the DAG.
We do the modification in the box because we want to start the synchronization from these one-to-many operators with the reconfiguration operators using epoch markers. The remaining steps are the same as in Algorithm~\ref{alg:hybrid-reconfig}.

As an example, in Figure~\ref{fig:extension-example-complex}, the only one-to-many operator is the {\sf Join} operator $J$. Because the reconfiguration operator $FMX$'s earliest ancestor one-to-many operator is $J$, we add $J$ to $M$ when constructing the MCS. The resulting MCS includes a single component with operators $J$, $SP$, and $FMX$, together with their edges. The controller injects an FCM to operator $J$, which propagates an epoch marker within the component to operator $FMX$.

We show that the extended \sysname scheduler still guarantees conflict-serializability of its produced schedule.

\begin{lemma}\label{lemma:hybrid-mcs-one-to-many}
(Corresponding to Lemma~\ref{lemma:mcs-component}.) Consider a dataflow graph $G$ with a reconfiguration $\mathcal{R}$, and the MCS $G'$ generated by Algorithm~\ref{alg:hybrid-reconfig-extended}.
Each component of $G'$ contains at least one reconfiguration operator.
\end{lemma}
\begin{proof}
Let $M$ be the set of operators used in Algorithm~\ref{alg:hybrid-reconfig-extended} to compute the MCS in line~\ref{alg:hybrid:line:mcs-init}. Using steps similar to those in Lemma~\ref{lemma:mcs-component}, we can show each component contains at least one operator in $M$.
By the construction in Algorithm~\ref{alg:hybrid-reconfig-extended}, an operator $P$ in $M$ is either (1) a reconfiguration operator, or (2) an earliest one-to-many ancestor operator of a reconfiguration operator $O$.
In the latter case, by the construction in Algorithm~\ref{alg:hybrid-reconfig-extended}, $O$ is also in $G'$. By the definition of components, $O$ is in the same component as $P$. Therefore, in both cases, each component of $G'$ contains at least one reconfiguration operator.
\end{proof}

\begin{lemma} 
\label{lemma:hybrid-component-extended}
(Corresponding to Lemma~\ref{lemma:hybrid-component}.) Consider a dataflow graph $G$ with a reconfiguration $\mathcal{R}$, and the MCS $G'$ generated by Algorithm~\ref{alg:hybrid-reconfig-extended}. For each source tuple, the operators in its data transaction overlap with at most one component of $G'$.
\end{lemma}

\begin{proof}
Suppose the operators processing a tuple $t$ overlap with two components $\mathcal{C}_1$ and $\mathcal{C}_2$ in the MCS. 
By Lemma~\ref{lemma:hybrid-mcs-one-to-many}, there is a reconfiguration operator $A$ in $\mathcal{C}_1$ and another reconfiguration $B$ in $\mathcal{C}_2$.
Let $(\mathcal{S}, \preccurlyeq_\mathcal{S})$ be the scope of $t$. Let $t_A$ and $t_B$ be two tuples (in $\mathcal{S}$) processed by operators $A$ and $B$, respectively.
Notice that the partial order $\preccurlyeq_\mathcal{S}$ of the scope forms a tree.
Let $t_L$ be the latest common ancestor tuple of $t_A$ and $t_B$ in the tree. By the definition of the scope $(\mathcal{S}, \preccurlyeq_\mathcal{S})$, the receiving operator $L$ of $t_L$ must be a one-to-many operator because there is more than one child of $t_L$ in the tree.
Notice that $L$ is a common ancestor of $A$ and $B$ because it has paths to both operators in $G$.
For operator $A$, by the construction in Algorithm~\ref{alg:hybrid-reconfig-extended}, $L$ is either (1) an earliest one-to-many operator of $A$, or (2) on the path between $A$ and an earliest one-to-many operator of $A$. In both cases, $L$ is in $G'$. 
By the definition of components, since $L$ is in $G'$ and $L$ is connected to both $A$ and $B$, $A$ and $B$ must be in the same component of $G'$, which contradicts the assumption.
\end{proof}

\begin{lemma} 
\label{lemma:hybrid-component-head}
Consider a dataflow graph $G$ with a reconfiguration $\mathcal{R}$, and the MCS $G'$ generated by Algorithm~\ref{alg:hybrid-reconfig-extended}. A head operator $H$ in a component of $G'$ receives at most one input tuple.
\end{lemma}
\begin{proof}
Let $M$ be the set of operators used in~\ref{alg:hybrid-reconfig-extended} to compute the MCS. Suppose a head operator $H$ of a component in $G'$ is not in $M$. We construct a new sub-DAG $G''$ by removing $S$ and its edges from $G'$. Using similar steps as in Lemma~\ref{lemma:mcs-unique}, we can show that $G''$ is still a minimal covering sub-DAG of $G$ and $M$, which contradicts the minimality property of $G'$.

By the construction in Algorithm~\ref{alg:hybrid-reconfig-extended}, operator $H$ is either (1) a reconfiguration operator with no ancestor one-to-many operators, or (2) an earliest one-to-many operator of a reconfiguration operator, which also has no ancestor one-to-many operators. Therefore, $H$ can receive at most one input tuple in $T$.
\end{proof}

\begin{theorem}
\label{theorem:hybrid-method-extended}
(Corresponding to Theorem~\ref{theorem:hybrid-method}.) For a workflow possibly with one-to-many operators and a reconfiguration request, Algorithm~\ref{alg:hybrid-reconfig-extended} always produces a conflict-serializable schedule.
\end{theorem}

\begin{proof}
Consider the function-update transaction $U$ of a reconfiguration $\mathcal{R}$ in the algorithm and the data transaction $T$ for a source tuple $t$. Lemma~\ref{lemma:hybrid-component-extended} shows that the operators in $T$ overlap with at most one component of the MCS $G'$ produced in the algorithm.
Consider a head operator $H$ in a component of $G'$. 
Lemma~\ref{lemma:hybrid-component-head} shows that a head operator $H$ in a component of $G'$ can receive at most one input tuple in $T$. 
We compare the position of the epoch marker on $H$ with a possible single input tuple of $H$ to determine the position of $T$ in a serial schedule. We can prove this claim using steps similar to those in the proof of Theorem~\ref{theorem:hybrid-method}.
\end{proof}

\subsection{Reducing delay by MCS pruning}
\label{subsec:one-to-many-pruning}

For dataflows with one-to-many operators, the reconfiguration delay can be long when there are many intermediate operators between the head of an MCS component and a reconfiguration operator in the component. 
To address this limitation, we improve the \sysname scheduler in Algorithm~\ref{alg:hybrid-reconfig-extended} by using pruning rules to remove one-to-many operators that do not need to be synchronized.
Algorithm~\ref{alg:hybrid-reconfig-pruning} shows the addition of a pruning step. In line~\ref{alg:hybrid-plus:line:prune}, we call a function {\sf pruneAncestors} that applies pruning rules to each of the ancestor one-to-many operators to decide it can be pruned.

\begin{algorithm}
\caption{The \sysname Scheduler with a Pruning Process}\label{alg:hybrid-reconfig-pruning}
\algrenewcommand\algorithmicrequire{\textbf{Input:}}
\algrenewcommand\algorithmicensure{\textbf{Output:}}
\begin{algorithmic}[1]
    \State $M = \{o_1, \ldots, o_n\}$
    \For {each reconfiguration operator $o_i$ in $\{o_1, \ldots, o_n\}$} \label{alg:hybrid-plus:line:one-to-many-start}
        \tikzmark{pa1}
        \State $\mathcal{A} \leftarrow$ set of ancestor one-to-many operators of $o_i$
        \State $\boldsymbol{pruneAncestors(\mathcal{A})}$         \label{alg:hybrid-plus:line:prune}
        \State $\mathcal{E} \leftarrow computeEarliestAncestors(\mathcal{A}$)
        \State $M \leftarrow M \cup \mathcal{E}$
        \label{alg:hybrid-plus:line:one-to-many-body}
    \tikzmark{pb2} \EndFor \label{alg:hybrid-plus:line:one-to-many-end}
    \State \ldots same as Algorithm~\ref{alg:hybrid-reconfig} line \ref{alg:hybrid:line:find-mcs-end}-\ref{alg:hybrid:line:epoch-end}
\end{algorithmic}
\begin{tikzpicture}[remember picture,overlay]
\draw[black,rounded corners]
  ([shift={(-205pt,2ex)}]pic cs:pa1) 
    rectangle 
  ([shift={(150pt,-0.5ex)}]pic cs:pb2);
\end{tikzpicture}
\end{algorithm}

Next, we introduce two pruning rules that are used in the improved \sysname scheduler.

\boldstart{1. Edge-wise one-to-one pruning rule.} Figure~\ref{fig:replicate-pruning}~(I) shows a part of a dataflow with a {\sf Replicate} operator, denoted as $RE$. This operator replicates each input tuple to produce two output tuples and sends each of them to operators $C$ and $D$. $RE$ is a one-to-many operator by Definition~\ref{def:one-to-many-op}. Suppose all other operators in this dataflow are one-to-one operators. Using Algorithm~\ref{alg:hybrid-reconfig-extended}, the \sysname scheduler includes operators $RE$, $C$, and $E$ in the MCS, as shown in the red box in Figure~\ref{fig:replicate-pruning}~(I). This is because $RE$ is the earliest one-to-many ancestor operator of the reconfiguration operator $E$.

\begin{figure}[htbp]
	\includegraphics[width=\linewidth]{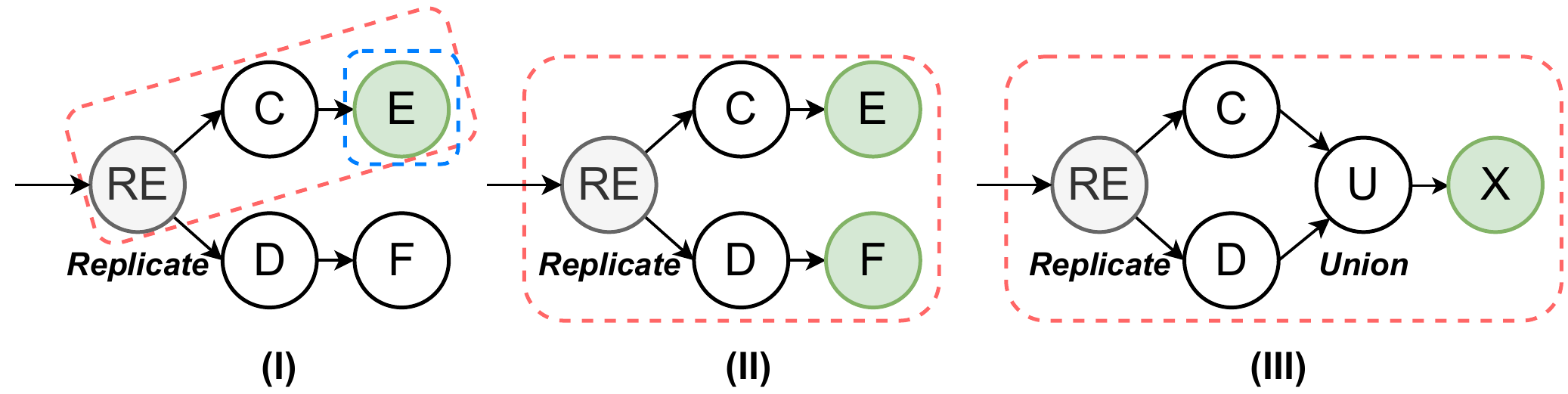}
	\caption{\label{fig:replicate-pruning}
		\textbf{Example reconfigurations on dataflows with a replicate operator. (I): The MCS can be pruned. (II) and (III): the MCS's cannot be pruned. }
	}
\end{figure}

Although operator $RE$ is a one-to-many operator, for an input tuple, the operator outputs a single tuple on each edge. For the reconfiguration operator $E$, it only receives a single tuple in each data transaction. Therefore, there is no need for operator $E$ to synchronize with operator $RE$. The MCS only contains operator $E$, as shown in the blue box in Figure~\ref{fig:replicate-pruning}. 
Figure~\ref{fig:replicate-pruning}~(II) and (III) show dataflows where the MCS with a replicate operator cannot be pruned.  In Figure~\ref{fig:replicate-pruning}~(II), for each tuple processed by operator $E$, the corresponding replicated tuple must be processed by the same version of operator $F$. We can achieve the goal by starting the synchronization from $RE$. In Figure~\ref{fig:replicate-pruning}~(III), operator $X$ receives all the replicated tuples in each data transaction. Therefore we also need to start the synchronization from the one-to-many operator $RE$.

Next, we formally describe the pruning rule.  We prune an ancestor one-to-many operator $A$ of a reconfiguration operator $o_i$ if the following conditions are true. (1) On each of its output edges, $A$ emits at most one tuple for each input tuple.  (2) The $A$ has only one output edge $e$ connected to a downstream reconfiguration operator, and this output edge $e$ is connected to $o_i$.
Intuitively, condition (1) ensures that $A$ behaves like a one-to-one operator on each of its output edges. 
Condition (2) ensures that the reconfiguration transaction of $o_i$ affects only one output tuple of $A$ sent on edge $e$.
As analyzed in Section~\ref{subsec:one-to-many-extension}, a one-to-many operator $O$ needs to be included in the MCS to ensure multiple output tuples of $O$ are processed using the same configuration. In this case, only a single output tuple of $A$ is affected by the reconfiguration. Therefore, $A$ can be pruned from the set of operators used to construct the MCS.

\boldstart{2. Uniqueness pruning rule.}

Next, we show another example of pruning an one-to-many operator. In Figure~\ref{fig:self-join-pruning}, suppose we want to reconfigure operator $E$. Each input tuple is first replicated by operator $RE$. The replicated tuples are sent to operators $C$ and $D$. They are then combined to a single tuple using a {\sf Self-Join} operator $SJ$ on the primary key. 
Algorithm~\ref{alg:hybrid-reconfig-extended} computes the sub-DAG from operator $RE$ to operator $E$ as the MCS, as shown in the red box in Figure~\ref{fig:self-join-pruning}. However, notice that operator $SJ$ ensures that it generates at most one output tuple for input tuple from the source. Therefore, $RE$ does not need to be synchronized and the MCS only needs to contain $E$ without $RE$, as shown in the blue box.
In general, we prune an ancestor one-to-many operator $A$ of a reconfiguration operator $o_i$ if on each path from $A$ to $o_i$, there exists an operator $O$ that has the following uniqueness property: operator $O$ generates at most one output tuple for each data transaction. In the running example, $SJ$ is such an $O$ operator and $RE$ can be pruned.

\begin{figure}[htbp]
	\includegraphics[width=2.5in]{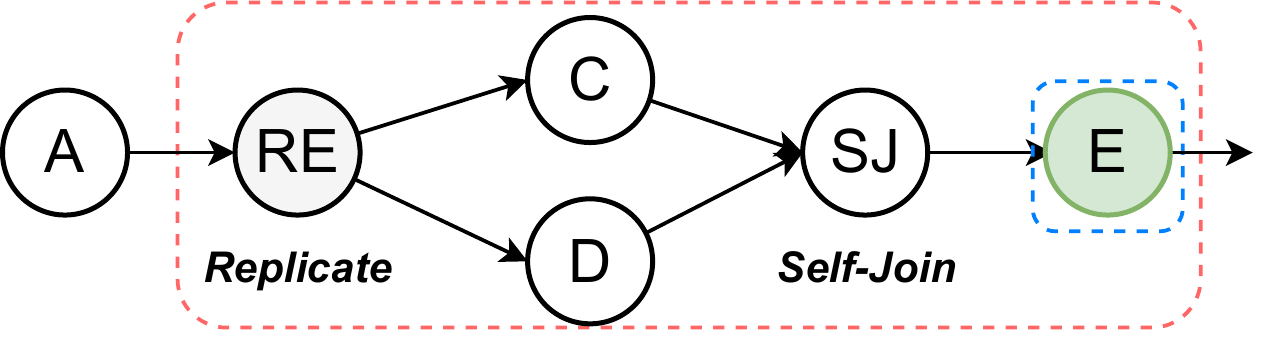}
	\caption{\label{fig:self-join-pruning}
		\textbf{Operator $RE$ can be pruned from the set of operators used to construct the MCS. }
	}
\end{figure}

\section{Extensions}
\label{sec:extensions}
In this section, we consider how the \sysname scheduler works in more general cases, including the case of workflows with blocking operators and the case where an operator has multiple workers. Moreover, we discuss how to support fault tolerance in the \sysname scheduler.

\subsection{Dataflows with Blocking Operators}
\label{subsec:extension-blocking}
We now consider how the \sysname scheduler works on dataflows containing blocking operators, such as aggregation and sort. Consider a blocking operator $B$. All operators before $B$ need to run to their completion before the operators after $B$ start to run.
In other words, the operators before $B$ and those after $B$ never execute at the same time.  Based on this observation, we can use the blocking operators in a dataflow to divide the dataflow into multiple sub-dataflows, with each of them containing pipelined operators only. Then we run \sysname on each sub-dataflow during its execution.

\subsection{Multiple Workers for an Operator}
\label{subsec:extension-parallel}

In a parallel execution engine, each operator can have multiple workers, with each worker processing a data partition. 
We map a single-worker dataflow $G=(V,E)$ to a parallel dataflow $G^*=(V^*,E^*)$, where each operator $v$ in $V$ is mapped to multiple parallel workers $v^1, \ldots, v^{p}$ in $V^*$, where $p$ is the number of workers of the operator. 
We map a reconfiguration $\mathcal{R}$ specified on the single worker dataflow $G$ to a new reconfiguration $\mathcal{R}^*$ on the parallel dataflow $G^*$ of $G$.
Figure~\ref{fig:parallel-example} shows part of a parallel dataflow based on Figure~\ref{fig:extension-example-complex}, where each operator runs using two workers.
For each function update $\mu(o_i)$ on an operator $o_i$ in $R$, we map it to a set of function updates on all the workers of $o_i$, i.e., $\{ (o^1_i,\mu(o_i)),\ldots, (o^{p}_i,\mu(o_i)) \}$ in $R^*$. For example, the reconfiguration on operator $FMX$ is mapped to a reconfiguration on the corresponding workers $FMX_1$ and $FMX_2$.

\begin{figure}[htbp]
	\includegraphics[width=2.8in]{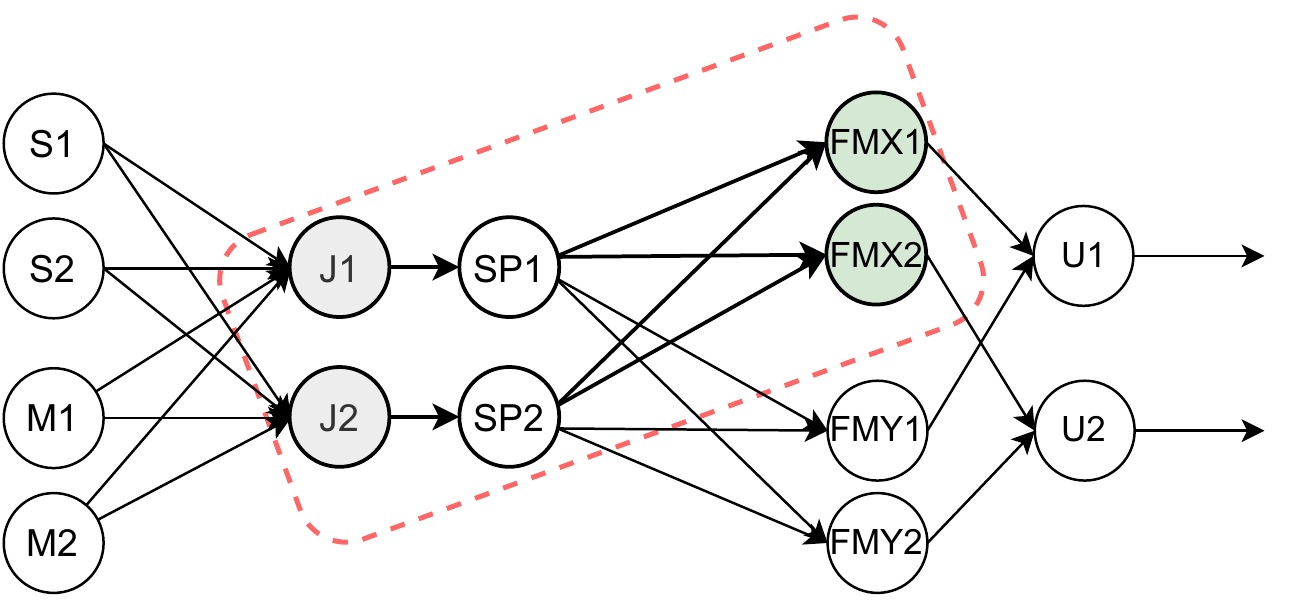}
	\caption{\label{fig:parallel-example}
		\textbf{A reconfiguration on a parallel dataflow with two workers per operator. }
	}
\end{figure}

Notice that the parallel dataflow $G^*$ is also a DAG. The \sysname scheduler in Algorithm~\ref{alg:hybrid-reconfig-pruning} can be directly run on $G^*$ with $\mathcal{R}^*$. 
The operators and edges in the generated MCS are highlighted in red in Figure~\ref{fig:extension-example-complex}
The \sysname scheduler treats a worker of an operator to have the same property (one-to-one or one-to-many) as the operator in hash and range partitioning.  For example, both workers of the {\sf Join} operator are treated as one-to-many operators. 
When using the broadcast strategy, a worker broadcasts an output tuple to all its downstream operators, same as the {\sf Replicate} operator described in Section~\ref{subsec:one-to-many-pruning}. In this case, the \sysname treats it as if a {\sf Replicate} operator is added after the worker. The pruning techniques described in Section~\ref{subsec:one-to-many-pruning} can still be used.

\subsection{Fault Tolerance Using the \sysname Scheduler}
\label{subsec:extension-fault-tolerance}

Fault tolerance requires that a system can recover to a consistent state in case of failures. For a dataflow $G$ and a reconfiguration $\mathcal{R}$, the execution of the dataflow is not in a consistent state if some operators in $\mathcal{R}$ are updated, and some operators in  $\mathcal{R}$ are not.
In an epoch-based scheduler, fault-tolerance can be supported using epoch-based checkpointing~\cite{journals/corr/CarboneFEHT15, journals/pvldb/CarboneEFHRT17}. 
However, such checkpointing cannot guarantee fault tolerance for the \sysname scheduler.
Consider the reconfiguration in Figure~\ref{fig:hybrid-min-subdag} and
the following sequence of events: (1) $G$ receives a checkpoint marker from $B$; (2) $G$ and $C$ receive the reconfiguration FCM's; and (3) $C$ receives a checkpoint marker from $A$. The checkpoint contains the old configuration of $G$ and the new configuration of $C$, which is not in a consistent state. Next we discuss two methods to support fault tolerance in \sysname.

\boldstart{Checkpoint-based fault tolerance.} 
When a reconfiguration request arrives at the controller, the controller cancels all in-flight checkpoints because they could produce inconsistent states. 
The controller then blocks any new checkpoints to be started until all head operators of each MCS component have received their FCM's.
In this way, the subsequent epoch markers will always be after the FCM's, thus the subsequent checkpoints only contain the fully updated configuration.
The blocking period is short because the FCM's are not blocked by any data messages.

\boldstart{Logging-based fault tolerance.} 
The FCMs introduce non-determinism in the execution of an operator. We can log all the non-determinism factors of each operator, including the arrival order of data tuples and the FCMs. During recovery, each operator is deterministically replayed and the FCMs are injected following the original order. We can leverage an existing logging-based fault-tolerance approach such as the one in Clonos~\cite{conf/sigmod/SilvestreFSK21}, which is built on top of Flink. FCMs can be modeled as RPC calls received by an operator in Clonos, which are recorded in the logs.

\section{Experiments}
\label{sec:experiments}

In this section, we present the results of experiments  of different reconfiguration schedulers and show the benefits of \sysname. 

\subsection{Setting}
\label{subsec:experiment-setting}

\boldstart{Datasets.} 
We used three datasets shown in Table~\ref{table:exp-datasets}. Dataset 1 had 24M tuples of credit card payments with 12 attributes~\cite{padhi2021tabular}, such as the customer, merchant, date, amount, and chip usage.
Dataset 2 was constructed by grouping the credit card payments per user in dataset 1.  Each record had a user and a list of payments by the user. We used this dataset to utilize a one-to-many {\sf unnest} operator to split a payment list into multiple records. 
Dataset 2 was generated using the TPC-DS benchmark~\cite{misc/tpcds} with a scale factor of 100.

\boldstart{Workflows.} We constructed workflows as shown in Figure~\ref{fig:exp-workflows}. Workflow ${W_1}$ simulated a fraud detection application, and it detected fraud of a user based on the user's historical payment amounts. By default, the source operator read the payment table with a rate of 1,000 tuple/s. The user-based inference operator saved 10 most recent payment amounts for each user as its internal state. For each input tuple, the operator updated the user's state and used an LSTM auto-encoder~\cite{wiese2009credit} to predict the probability of fraud. Workflow ${W_2}$ was constructed based on TPC-DS query 40. For all items with a price between 0.99 and 1.49, this workflow computed the item id and location of the warehouse the item was delivered from in a 60-day period. 
Workflow ${W_3}$ was constructed based on TPC-DS query 71. It produced the brands managed by a given manager that sold their products across three sales channels at either breakfast or dinner time for a given month.
All the join operators in these workflows were one-to-one operators because they join a primary key with a foreign key. We only considered the pipelined sub-DAG of each dataflow. For example, if a hash join has a build phase and a probe phase, we only consider the pipelined probe phase. In Figure~\ref{fig:exp-workflows}, we highlighted all the pipelined edges considered in the experiment in red.

On top of ${W_1}$, workflow ${W_4}$ included an additional merchant-based inference operator for users who made a large amount of payments. For each merchant, the inference operator saved 50 most recent payments, and used a similar LSTM auto-encoder to do inference. A payment record was processed by both inference operators. If one of them produced a probability greater than a threshold, the payment record was flagged as fraud. Workflow ${W_5}$ replicates the payment tuples to both the user-based inference operator and the merchant-based operator. After each operator makes fraud predictions, a {\sf Self-Join} operator is used to combine the replicated tuples into a single one.

\begin{table}[htbp]
\centering
\resizebox{\linewidth}{!}{
\begin{tabular}{|c|l|r|r|}
\hline
Dataset & Table                 & Attribute \# & Tuple \#\\ \hline
1 & Credit card payment           &12           &24M      \\ \hline
2 & Credit card payment aggregated per user  & 2           & 20K\\ \hline
\multirow{3}{*}{3} & Catalog sales &34 & 144M\\ \cline{2-4}
                  & Store sales & 23 &288M \\ \cline{2-4}
                  & Web sales & 34 &71M \\ \hline
\end{tabular}
}
\caption{Datasets used in the experiments.}
\label{table:exp-datasets}
\end{table}

\begin{figure}[hbtp]
	\includegraphics[width=\columnwidth]{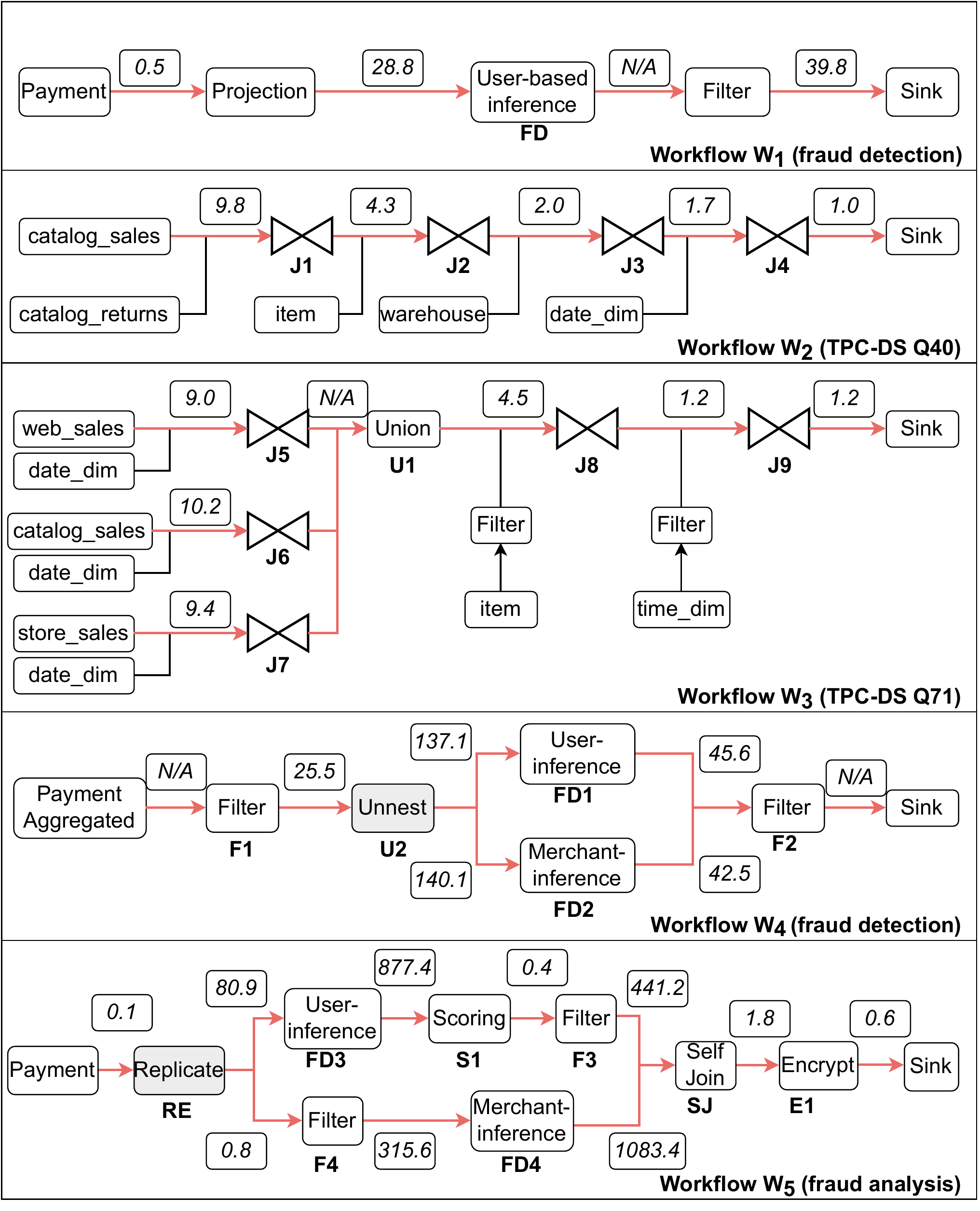} 
	\caption{\label{fig:exp-workflows}
	\textbf{Workflows used in the experiments. Pipelined edges are highlighted in red.} }
\end{figure}
 
\boldstart{Reconfigurations.}
For workflow $W_1$, we performed configurations with one operator.  For workflows $W_2$, $W_3$, and $W_4$, we performed reconfigurations with multiple operators. 
The methods of choosing reconfiguration operators will be described in each experiment.

\boldstart{Schedulers.}  We implemented two epoch-based schedulers. The first one performed a reconfiguration with a savepoint, which was natively supported by Flink (described in Section~\ref{sec:epoch-reconfiguration}). The second one was the scheduler of Chi~\cite{journals/pvldb/MaiZPXSVCKMKDR18} (described in Section~\ref{sec:epoch-reconfiguration}). As Chi was not open source, we implemented this scheduler on top of Flink, and used Flink's aligned checkpoint barriers as epoch markers. The first scheduler always stopped and restarted the execution {\em after} the propagation of checkpoint barriers to apply reconfiguration. The second scheduler applied reconfiguration {\em during} the barrier propagation, and did not require an additional stop-and-restart of the system. As a consequence, the second scheduler always had a shorter reconfiguration delay than the first, as verified in our experiments. Therefore, between these two schedulers, we only report the results of the second, denoted as ``\epoch scheduler.'' 

For fair-comparison purposes, we implemented \sysname also on top of Flink. In the implementation, FCM's between the controller and a specific worker of an operator were sent in special network channels (available in Flink). For each MCS $C$ computed in \sysname for a reconfiguration, the controller sent FCM's to the workers of $C$'s head operators. These workers pushed checkpoint barriers to the workers of their downstream operators in $C$. To let every operator know which downstream operators were in $C$, the checkpoint barrier also included $C$ and the reconfiguration operators in $C$. Every reconfiguration operator in $C$ applied the reconfiguration after receiving checkpoint barriers from all its upstream operators in $C$. The reconfiguration for this MCS $C$ completed after all $C$'s reconfiguration operators applied the reconfiguration. 

\boldstart{System environment.} 
All the experiments were conducted on the Google Cloud Platform (GCP). The execution was on a GCP dataproc cluster with 1 coordinator machine and 10 worker machines. All the machines were of type n1-highmem-4 with Ubuntu 18.04. The job controller of Flink ran on the coordinator. The coordinator machine had a 2TB HDD, while each worker machine had a 250GB HDD. To separate computation and storage, we stored the datasets in an HDFS file system on another cluster with 6 e2-highmem-4 machines, each with 4 vCPU’s, 32 GB memory, and a 500GB HDD. For all the schedulers, we used Flink release 1.13 and Java 8.

\subsection{Choke Point Analysis of Workflows}
In the execution there were various choke points in the workflow where the reconfiguration delay between two operators was very high.
We analyzed these choke points in the experiment workflows by computing the average reconfiguration delay between two operators using the epoch scheduler and showed the numbers on top of each edge in Figure~\ref{fig:exp-workflows}. The numbers represented the delay from the time when the upstream operator applied the reconfiguration and sent checkpoint barriers to the time when the downstream operator aligned all the checkpoint barriers and applied the reconfiguration. Some edges are marked as {\sf N/A} because the two connected operators were fused to a single operator chain in Flink. Edges marked with a number perform re-partition operations, thus the two connected operators are not chained.

We had the following observations. 1) Expensive operators usually created choke points in the workflow. For example, in $W4$, both inference operators applied the reconfiguration from $U1$ after the checkpoint barriers were sent out for around 140s. The inference operators accumulated input tuples in their input data channel, which blocked the checkpoint barrier to be processed. 2) Stragglers also created choke points. For example, in $W5$, there was a delay of 877.4s between $FD3$ and $S1$, because one of the $FD3$ workers was a straggler. Recall that due to the epoch alignment step, $S1$ had to receive all the checkpoint barriers before applying the reconfiguration. $S1$ was blocked when waiting for the straggler $FD3$ worker to finish.
3) If operators had similar costs, choke points depended on the amount of data in each operator's input data channel. For example, in both $W2$ and $W3$, the first several joins had larger delays of reconfiguration. This is because the data was filtered by every join, and the joins near the sink received less data so they had a lower reconfiguration delay.

\subsection{Benefits of Short Reconfiguration Delay: Reducing End-to-end Tuple Latency}

A main advantage of \sysname was its short reconfiguration delay compared to epoch-based schedulers. To show the benefits of this advantage, we considered a scenario for $W_1$ as shown in Figure~\ref{fig:throughput-compare}, where the developer needed to hot-replace the model in the user-based inference operator $FD$ during the execution to deal with a sudden surge of input data. In $W_1$, we set the number of workers for operators (except for the source and sink) to 40 to utilize all the cores in the cluster. The time for $FD$ to process a tuple was about 25ms, and the maximum throughput of this operator was around $1,600$ tuple/s. We used a single worker of the source operator and another single worker of the sink operator on the same machine so that they can use the same clock. The source operator started with an initial ingestion rate of $1,000$ tuples/s. After 100 seconds, we increased the ingestion rate to $2,000$ tuples/s. The developer saw an increasing trend of the end-to-end tuple latency. He requested a reconfiguration to replace the original LSTM auto-encoder model in FD with another LSTM auto-encoder with fewer parameters at $t = 120s$ to speed up the processing. The developer continuously monitored the end-to-end tuple latency. After another $100$ seconds, we further increased the rate to $9,000$ tuples/s. The developer decided to further decrease the cost of $FD$ to reduce the latency. At $t = 220s$, he requested another reconfiguration to replace the LSTM auto-encoder model with a simple decision-tree model.

\begin{figure}[htbp]
	\includegraphics[width=\columnwidth]{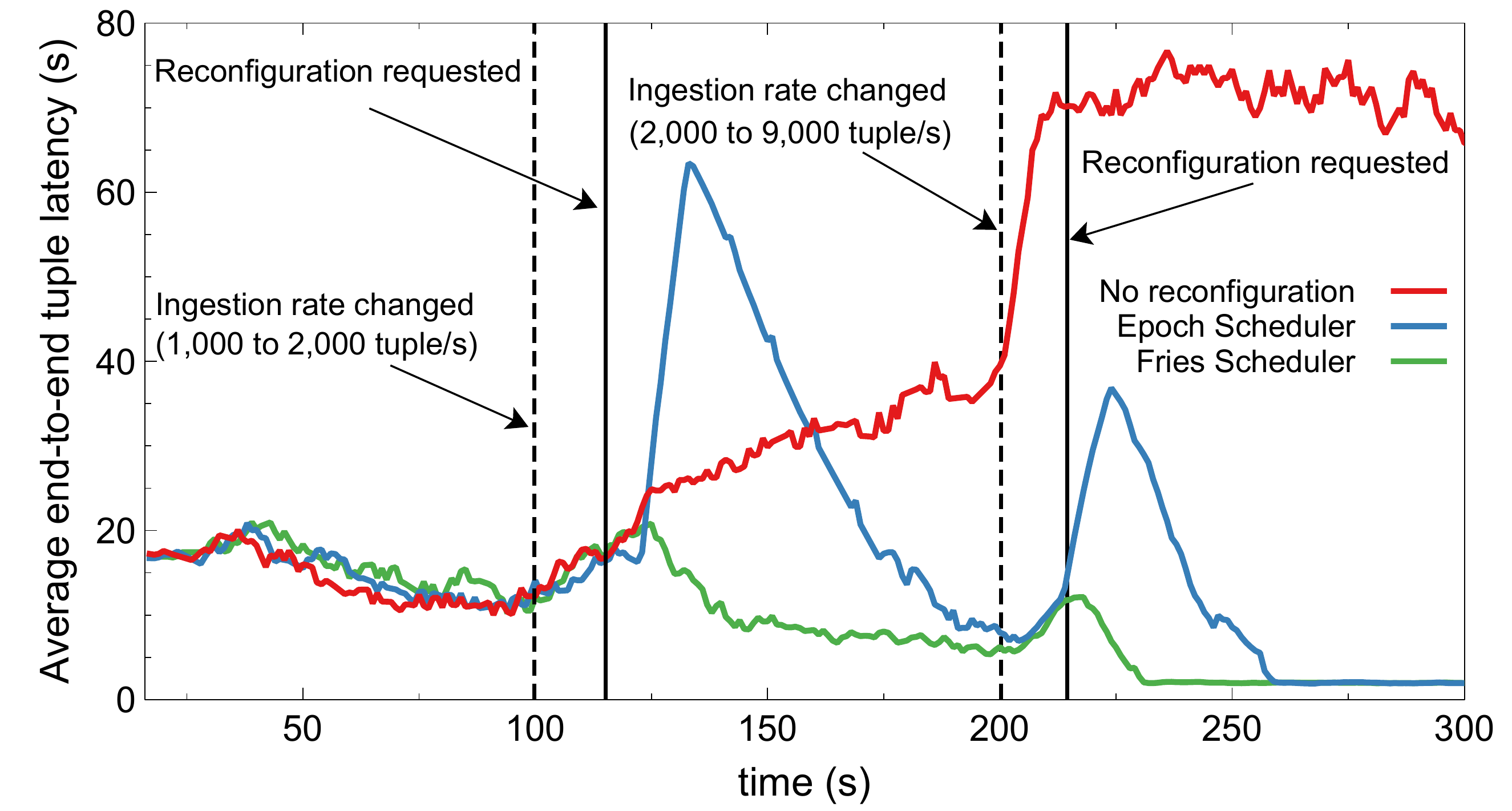} 
	\caption{\label{fig:throughput-compare}
	\textbf{Effect of mitigating surges of data-ingestion rate by different schedulers ($W_1$ on dataset 1).}}
\end{figure}

To compute the end-to-end latency of each input tuple, we attached a timestamp at the moment when the source operator generated this tuple. When the tuple was received by the sink operator, the latency was computed as the difference between the current time and the attached timestamp. Figure~\ref{fig:throughput-compare} shows the average end-to-end latency of output tuples received for every 10-second sliding window. (1) For the case without reconfiguration, after 100 seconds, the end-to-end latency began to increase because $FD$ was not fast enough to process all the incoming tuples. Many tuples were buffered in the network channel. The latency increased continuously and stabilized at around 70s when backpressure from $FD$ was propagated to the source operator to slow down the data ingestion rate.
(2) For the case of using the \epoch scheduler, the latency rapidly increased to above 60 seconds until around $t = 135s$ due to the surge. The main reason for the increase in the end-to-end latency was the blocking in the epoch alignment step. 
Since the sink operator had only one worker, it needed to wait until all 40 upstream $FD$ workers completely processed all tuples in the old epoch before it could process any tuples in the new epoch. Note that the delay was determined by the slowest $FD$ worker. We observed that the average delay of all workers was 30 seconds. However, there were two straggler workers that took 58 seconds and 69 seconds to finish processing the old epoch, respectively. The two straggler workers suffered from data skew. On average, each worker processed 35,000 tuples in the old epoch. However, the slowest worker processed 62,000 tuples.

(3) For the case of using the \sysname scheduler, the latency immediately decreased after $t = 120s$, indicating that $FD$ applied the reconfiguration and was able to quickly process the buffered tuples. Compared to \epoch, \sysname required less time to mitigate the surge. In this reconfiguration, the MCS component contained operator $FD$ only. Therefore, FCMs are directly sent to all $FD$ workers and no epoch markers were propagated to any downstream operators. This eliminated the aforementioned delay caused by the epoch alignment step.

\subsection{Benefits of Short Reconfiguration Delay: Reducing Wasted Computing Resources}

Another benefit of a short reconfiguration delay can be illustrated in the following scenario.
When processing data with unexpected content or formats, the workflow can produce invalid output tuples to be collected and reprocessed. A large number of invalid output tuples not only wastes computation in the current execution, but also requires more resources in the future. A short reconfiguration delay can effectively reduce the amount of wasted computing resources. To illustrate the benefit, we considered workflow $W_1$, and attached a version number $V1$ to every source tuple. The user-based inference operator $FD$ had another version number $V2$ for its processing logic. For every input tuple, $FD$ expected $V1$ of the tuple to match with $V2$; otherwise, the operator produced an invalid output tuple. For every 50 seconds, we increased $V1$, and the developer realized the version changed 20 seconds after that. Then, he requested a reconfiguration of $FD$ to also increase its version $V2$. We measured the number of invalid output tuples produced by the workflow over time under different reconfiguration schedulers as a metric of wasted computing resources.

\begin{figure}[htbp]
	\includegraphics[width=\columnwidth]{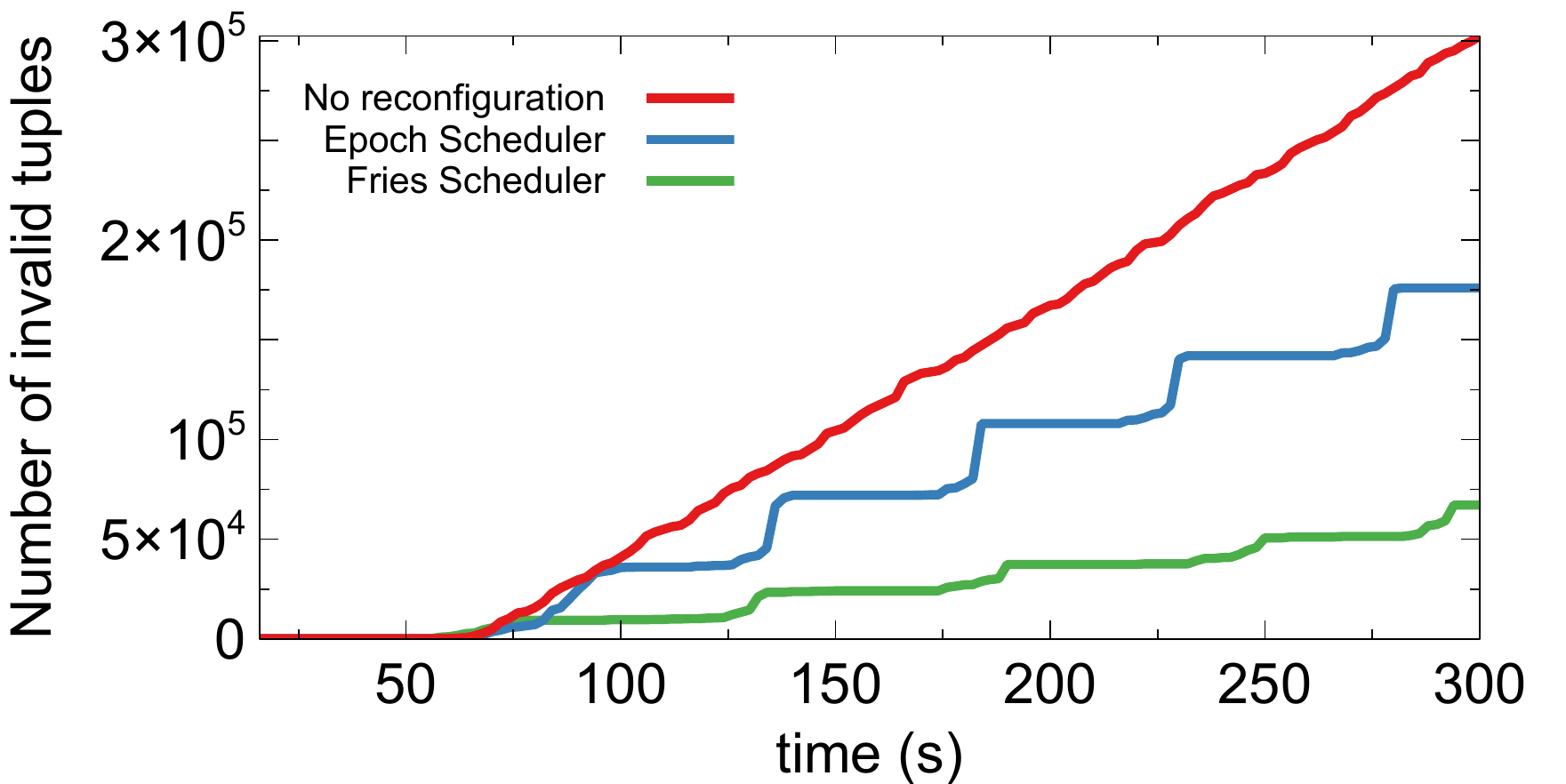} 
	\caption{\label{fig:wrong-tuples}
	\textbf{Effect of schedulers on the number of invalid output tuples ($W_1$ on dataset 1).}}
\end{figure}

Figure~\ref{fig:wrong-tuples} showed the result. (1) For the case without reconfiguration, all the output tuples after the first input version update were invalid. Thus the number of invalid tuples increased continuously. After 300 seconds, the number of invalid tuples reached 302K. (2) For the case of using the \epoch scheduler, operator $FD$ had to process all the tuples prior to the epoch marker before applying the reconfiguration. So all the tuples after the input version update and prior to the marker were invalid. The workflow produced many invalid tuples after each input version update, resulting in reprocessing of 176K tuples after 300 seconds.(3) For the case of using the \sysname scheduler, the operator quickly applied the reconfiguration, so fewer invalid tuples were generated for each input version change. At the end, only 67K tuples needed to be reprocessed.

\subsection{Effect of Data Ingestion Rates on Reconfiguration Delays}

Next we evaluated the effect of different factors on the delay.  We first considered data-ingestion rate. For workflow $W_1$, we gradually increased the rate at the source operator from $500$ tuples/s to $2,500$ tuples/s. For each configuration, after the execution of $120$ seconds, we applied a dummy reconfiguration on operator $FD$ and measured the delay under the two schedulers. As shown in Figure~\ref{fig:epoch-delay-input-frequency} (with a log scale for the $y$-axis), when the ingestion rate increased, the delay of the \epoch scheduler also increased due to the larger amount of in-flight tuples prior to the epoch marker. Since the \sysname scheduler sent FCM's directly to $FD$, its delay grew much slower than the \epoch scheduler. 

\begin{figure}[htbp!]
	\includegraphics[width=0.8\linewidth]{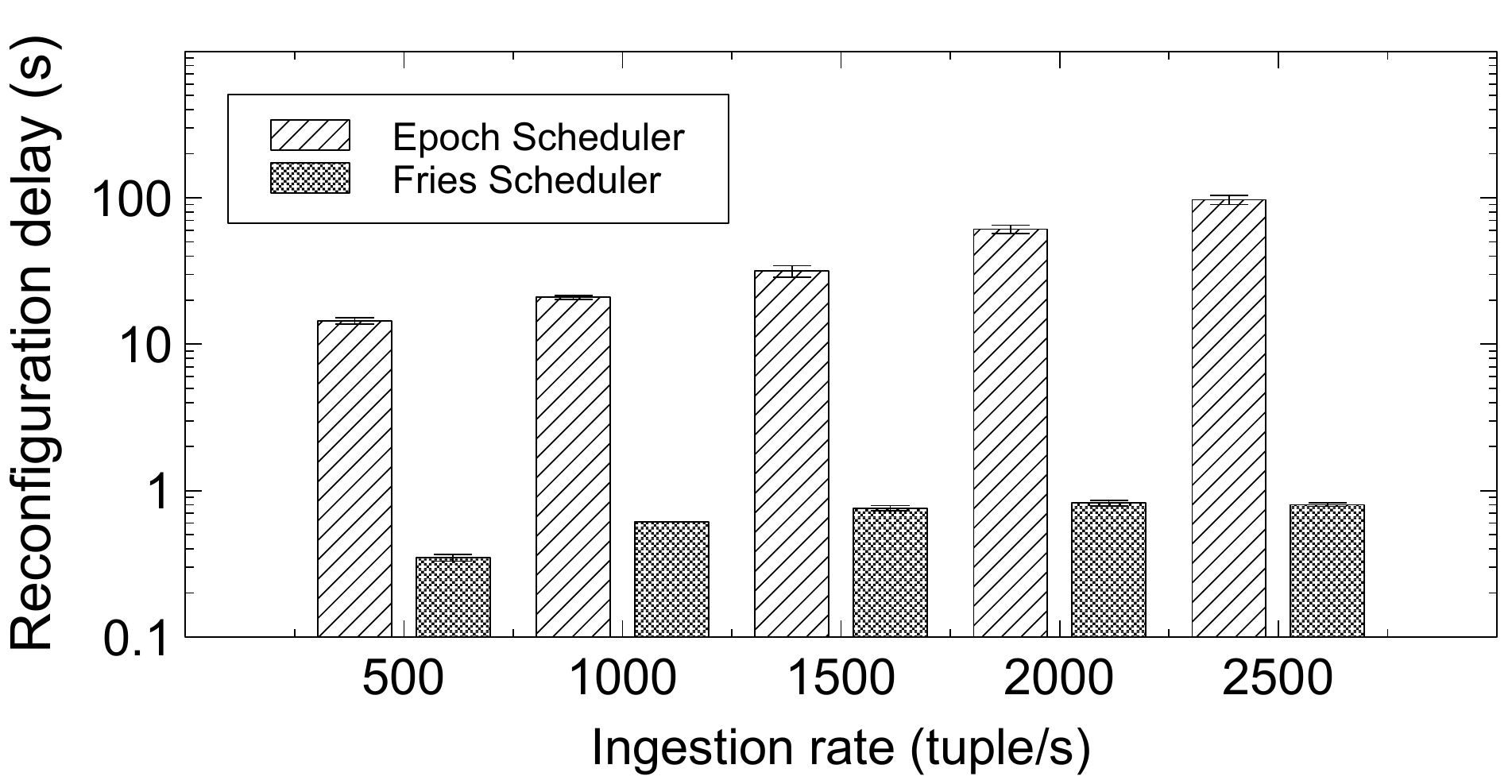} 
	\caption{\label{fig:epoch-delay-input-frequency}
	\textbf{Effect of data ingestion rate on the reconfiguration delay (with 95\% confidence intervals) ($W_1$ on dataset 1).}}
\end{figure}

\subsection{Effect of Operator Costs on Delays} 

To evaluate the effect of operator cost on the reconfiguration delay, for workflow $W_1$, we gradually increased the cost of the user-based inference operator $FD$ to process each input tuple. The $FD$ operator maintained a bounded queue of recent payment amounts of each user. When an input tuple was received by $FD$, the operator passed the payment amounts in the queue to its ML model. In different runs of experiments, we gradually increased the size of this queue from 10 to 50 so that the operator took more time to process each input tuple. Again, for each configuration, after the execution ran for $120$ seconds, we applied a dummy reconfiguration on $FD$ and measured the delay under the two schedulers. As shown in Figure~\ref{fig:epoch-delay-operator-cost}, when the $FD$'s cost increased, the delay of the \epoch scheduler also increased because each in-flight data tuple prior to the epoch marker took more time to be processed. On the other hand, the delay of the \sysname scheduler grew much slower than the \epoch scheduler. 

\begin{figure}[htbp]
	\includegraphics[width=0.8\linewidth]{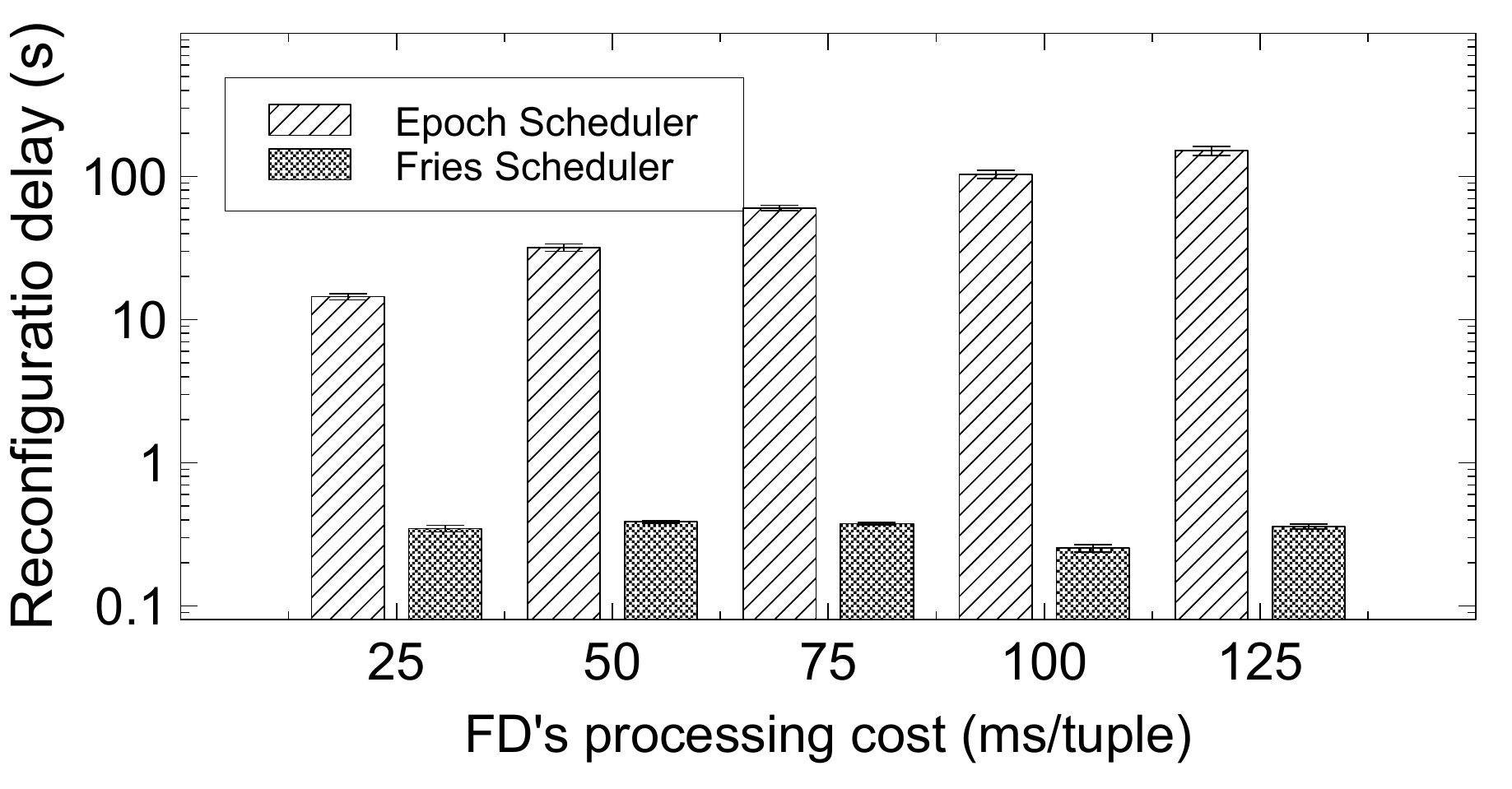} 
	\caption{\label{fig:epoch-delay-operator-cost}
	\textbf{Effect of operator cost on the reconfiguration delay with 95\% confidence intervals ($W_1$ on dataset 1).}}
\end{figure}

\subsection{Effect of Reconfigurations on Delays}

We wanted to evaluate the effect of reconfigurations on the delay under the two schedulers. We varied the number of reconfiguration operators in both workflows $W_2$ and $W_3$. For both workflows, we used 40 workers for each operator. For every 10 seconds, we requested a reconfiguration and measured the average reconfiguration delay. The results are shown in Table~\ref{table:reconf-op-and-mcs}. For each reconfiguration, we show its operators, the MCS components generated by the \sysname scheduler, the length of a longest path of each component, the delay of using the \sysname scheduler, and the delay of using the \epoch scheduler. We reported the path length because it affected the delay in the \sysname scheduler.

We have the following observations from the results. (1) The delay of the \sysname scheduler was always significantly lower than the delay of the \epoch scheduler. For example, for the $W_2$ configuration including $J1$ and $J4$, the delay of the \sysname scheduler was $1,702$ms, compared to $12,361$ms of the \epoch scheduler.
(2) For reconfigurations with multiple operators, if each operator formed a component in MCS, the delay of \sysname was very low. For example, for the $W_3$ reconfiguration of $J5$ and $J6$, each of them formed their own component. The \sysname scheduler had a delay of 127ms, comparable to 87ms in the case with $J5$ as the only reconfiguration operator. This low delay was because the \sysname scheduler sent FCM's separately to both operators and their reconfiguration happened in parallel.
(3) When the length of the longest path in a component increased, as expected, the delay of the \sysname scheduler also increased. For example, for the reconfiguration of $J1$ and $J3$, the longest path in their MCS had a length of 2, and the delay was $1,664$ms. For the reconfiguration of $J1$ and $J4$, the longest path in their MCS had a length of 3, and the delay increased to $1,702$ms.

\begin{table}[htbp]
\centering
\resizebox{\linewidth}{!}{
\begin{tabular}{|m{0.14\linewidth}|m{0.24\linewidth}|m{0.2\linewidth}|>{\raggedleft\arraybackslash}m{0.12\linewidth}|>{\raggedleft\arraybackslash}m{0.15\linewidth}|>{\raggedleft\arraybackslash}m{0.15\linewidth}|}
\hline
Workflow& Reconfiguration operators & MCS components & Longest path length & \sysname Scheduler delay (ms) & \epoch Scheduler delay (ms) \\ \hline
\centering{\multirow{4}{*}{$W_2$}}& J1       & \{\textbf{J1}\}              &0   &46     &11,432       \\ \cline{2-6}
& J2   & \{\textbf{J2}\}          &0   &44    &11,709         \\ \cline{2-6}
& J1, J3   & \{\textbf{J1}, J2, J3\}      &2   &1,664   &12,339    \\ \cline{2-6}
&J1, J4   & \{\textbf{J1}, J2, J3, J4\}  &3   &1,702   &12,361          \\ \cline{2-6} 
&J3, J4   & \{\textbf{J3}, J4\}  &1   &387   &13,767  \\ \hline
\centering{\multirow{5}{*}{$W_3$}}&J5                 & \{\textbf{J5}\}    &0  &87 &4,127      \\ \cline{2-6}
&\multirow{2}{*}{J5, J6}             & \{\textbf{J5}\} &0  &\multirow{2}{*}{127}  &\multirow{2}{*}{8,352}           \\
& &\{\textbf{J6}\} &0 & & \\ \cline{2-6}
&J5, J6, J7, J8     & \{\textbf{J5}, \textbf{J6}, \textbf{J7}, U1, J8\} &3 &447   &19,608          \\ \cline{2-6}
&J5, J6, J7, J9    & \{\textbf{J5}, \textbf{J6}, \textbf{J7}, U1, J8, J9\} &4  &526 &19,717         \\ \cline{2-6}
&J7, J8, J9    & \{\textbf{J7}, U1, J8, J9\} &3 &1,340   &20,532          \\ \hline
\end{tabular}
}
\caption{Reconfiguration operators, corresponding MCS, and reconfiguration delay in $W_2$ and $W_3$ on dataset 3. Head operators in each component are highlighted in bold.}
\label{table:reconf-op-and-mcs}
\end{table}

\subsection{Reconfiguration Delays in Workflows with One-to-many Operators}

We used workflow $W_4$ to evaluate the effect of reconfiguration operators on the reconfiguration delay in workflows with a one-to-many operator $U2$. This operator split all the payments of a user and sent them to both $FD1$ and $FD2$. Table~\ref{table:one-to-many-exp} shows the results. We have the following observations. 
(1) The delay of the \sysname scheduler was still always lower than the \epoch scheduler. 
(2) The reconfiguration of $FD1$ took a long time (47,892ms) in \sysname because $FD1$ was not the head operator of its component. The epoch markers had to go through the data channels of $FD1$ (from multiple workers). Since $FD1$ processed tuples slowly, many of its input tuples were buffered in its data channels, which delayed the propagation of the epoch markers.
(3) The reconfiguration of $F2$ took a long delay ($221,353$ms) in \sysname because its generated MCS contained every operator on the path from $U2$ and $F2$ with the  one-to-many $U2$ operator and both $FD1$ and $FD2$ were slow.

\begin{table}[htbp]
\centering
\resizebox{\linewidth}{!}{
\begin{tabular}{|m{0.23\linewidth}|m{0.22\linewidth}|>{\raggedleft\arraybackslash}m{0.15\linewidth}|>{\raggedleft\arraybackslash}m{0.2\linewidth}|>{\raggedleft\arraybackslash}m{0.2\linewidth}|}
\hline
Reconfiguration operators & MCS components & Longest path length & \sysname Scheduler delay (ms) & \epoch Scheduler delay (ms) \\ \hline
F1, U2               & \{\textbf{F1}, U2\} &1   & 69   &151  \\ \hline
FD1                  & \{\textbf{U2}, FD1\} &1  & 47,892  &  131,103   \\ \hline
F2               & \{\textbf{U2}, FD1, FD2, F2\} &5 &221,353 &   236,153        \\ \hline
\end{tabular}
}
\caption{Reconfiguration operators, corresponding MCS, and reconfiguration delay in $W_4$ on dataset 2. Head operators in each component are highlighted in bold.}
\label{table:one-to-many-exp}
\end{table}

\subsection{Effect of MCS Pruning on Delays in Workflows with One-to-many Operators}

We used workflow $W_5$ to evaluate the effect of the MCS pruning method proposed in Section~\ref{subsec:one-to-many-pruning} on the reconfiguration delay in workflows with a one-to-many {\sf Replicate} operator and a {\sf Self Join} operator. For each reconfiguration, we compare the \sysname scheduler with the pruning step turned on and turned off.
Table~\ref{table:one-to-many-pruning-exp} shows the results. We have the following observations. 
(1) In general, when pruning is possible, the size of MCS components was reduced and the delay with pruning was significantly lower than the delay without pruning. For example, the reconfiguration of operator $FD4$ and the reconfiguration of operator $F3$ benefited from the edge-wise one-to-one pruning rule.
(2) In the case of reconfiguring both $FD3$ and $FD4$, the pruning rules could not prune the one-to-many {\sf Replicate} operator. Therefore the delays were similar.
(3) The reconfiguration of operator $E1$ benefited from the uniqueness pruning rule. This reconfiguration had the largest benefit in delay because the number of edges in the MCS reduced from eight to zero, which greatly reduced the synchronization time.

\begin{table}
\centering
\resizebox{\linewidth}{!}{%
\begin{tabular}{|>{\hspace{0pt}}m{0.15\linewidth}|>{\hspace{0pt}}m{0.28\linewidth}|>{\hspace{0pt}}m{0.28\linewidth}|>{\hspace{0pt}}m{0.16\linewidth}|>{\hspace{0pt}}m{0.20\linewidth}|} \hline
Reconfiguration \par{}operators & MCS \par{} with pruning & MCS\par{}without pruning & Fries with pruning delay (ms) & Fries without pruning delay (ms) \\ \hline
FD4 & \{\textbf{FD4}\} & \{\textbf{RE}, F4, FD4\} & 158 & 450,149 \\ \hline
F3 & \{\textbf{F3}\} & \{\textbf{RE}, FD3, S1, F3\} & 94 & 383,781 \\ \hline
F4 & \{\textbf{F4}\} & \{\textbf{RE}, F4\} & 10 & 446 \\ \hline
FD3, FD4 & \{\textbf{RE}, FD3, F4, FD4\} & \{\textbf{RE}, FD3, F4, FD4\} & 661,892 & 663,460 \\ \hline
E1 & \{\textbf{E1}\} & \{\textbf{RE}, FD3, S1, F3,\par{}F4, FD4, SJ, E1\} & 85 & 1,122,686 \\ \hline
\end{tabular}
}
\caption{The effect of MCS pruning on delays in $W_5$.}
\label{table:one-to-many-pruning-exp}
\end{table}

\subsection{Effect of Multiple Workers on Delays}

To evaluate the effect of the worker number per operator on the reconfiguration delay, we considered workflow $W_2$ and increased the worker number per operator from 1 to 40. After the workflow ran for 20 seconds, we requested a dummy reconfiguration of $J1$ and $J4$. We measured the reconfiguration delay of the two schedulers.

\begin{figure}[htbp]
	\includegraphics[width=0.9\columnwidth]{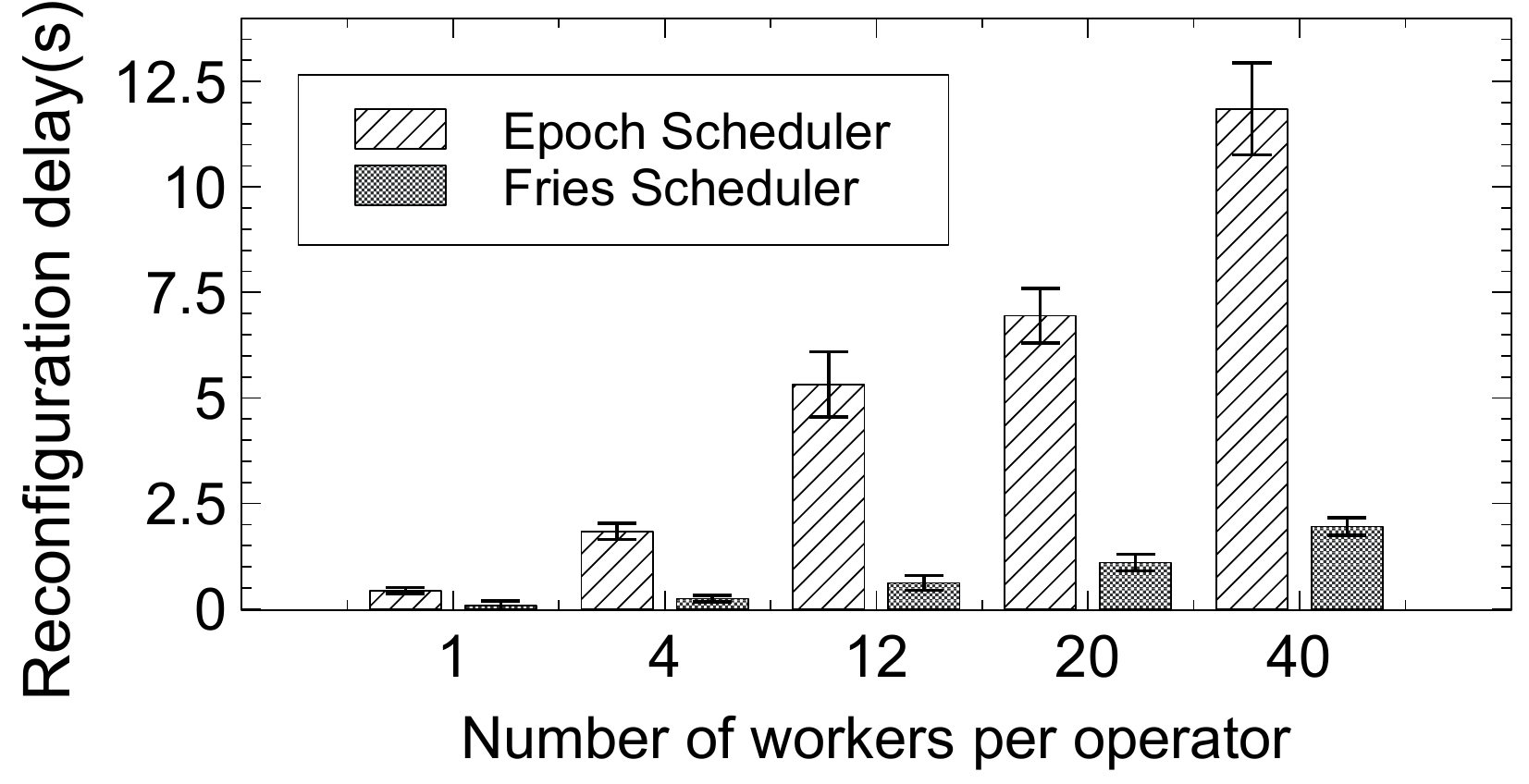} 
	\caption{\label{fig:parallelism-delay}
	\textbf{Effect of worker number on reconfiguration delay with 95\% confidence intervals ($W_2$ on dataset 3).}}
\end{figure}

As shown in Figure~\ref{fig:parallelism-delay}, as the worker number increased, the delay increased for both schedulers. This was because between each pair of join operators, the data was shuffled and every join worker needed to receive an epoch marker from all its upstream workers. So the number of data channels between the joins was the product of their numbers of workers. When each worker number increased, the number of epoch markers to collect also increased. 
The fact that the delay of \sysname scheduler was again lower than the \epoch scheduler can be explained using Table~\ref{table:parallelism-and-mcs-w2}.  In particular, the \sysname scheduler propagated epoch markers only through the data channels between MCS workers, while the \epoch scheduler had to propagate epoch markers through all the data channels. The table shows that the number of channels between MCS workers was always less than the number of channels between all workers.

\begin{table}[htbp]
\centering
\resizebox{\linewidth}{!}{
\begin{tabular}{|>{\raggedleft\arraybackslash}m{0.2\linewidth}|>{\raggedleft\arraybackslash}m{0.4\linewidth}|>{\raggedleft\arraybackslash}m{0.4\linewidth}|}
\hline
Worker \# per operator  & Total \# of data channels between all workers & Total \# of data channels between MCS workers \\ \hline
1           &  5  & 3       \\ \hline
4           & 68      & 48  \\ \hline
12        & 588     & 432     \\ \hline
20        &  1,620    & 1,200     \\ \hline
40        &  6,440    & 4,800    \\ \hline
\end{tabular}
}
\caption{Effect of number of workers on data channels for reconfiguration of $J1$ and $J4$ ($W_2$ on dataset 3).}
\label{table:parallelism-and-mcs-w2}
\end{table}

\section{Conclusions}

In this paper we studied the problem of runtime configurations in data-intensive workflow systems with a low delay. We showed limitations of existing epoch-based reconfiguration schedulers on the delay.  We developed a new technique called \sysname that uses fast control messages to do reconfigurations.  We formally defined consistency in runtime reconfigurations, and developed a \sysname scheduler with consistency guarantee.  The technique also works for parallel executions and supports fault tolerance. Our extensive experimental evaluation showed the advantages of this technique compared to epoch-based schedulers.

\boldstart{Acknowledgements:} This work was supported by the NSF IIS-2107150 award. We thank Sadeem Alsudais and Yicong Huang for their participation in discussions.

\bibliographystyle{ACM-Reference-Format}
\bibliography{references}


\begin{thebibliography}{39}

%%% the Web address.

\ifx \showCODEN    \undefined \def \showCODEN     #1{\unskip}     \fi
\ifx \showDOI      \undefined \def \showDOI       #1{#1}\fi
\ifx \showISBNx    \undefined \def \showISBNx     #1{\unskip}     \fi
\ifx \showISBNxiii \undefined \def \showISBNxiii  #1{\unskip}     \fi
\ifx \showISSN     \undefined \def \showISSN      #1{\unskip}     \fi
\ifx \showLCCN     \undefined \def \showLCCN      #1{\unskip}     \fi
\ifx \shownote     \undefined \def \shownote      #1{#1}          \fi
\ifx \showarticletitle \undefined \def \showarticletitle #1{#1}   \fi
\ifx \showURL      \undefined \def \showURL       {\relax}        \fi
% invisible to TeX
\providecommand\bibfield[2]{#2}
\providecommand\bibinfo[2]{#2}
\providecommand\natexlab[1]{#1}
\providecommand\showeprint[2][]{arXiv:#2}

\bibitem[\protect\citeauthoryear{Akidau, Bradshaw, Chambers, Chernyak,
  Fern{\'{a}}ndez{-}Moctezuma, Lax, McVeety, Mills, Perry, Schmidt, and
  Whittle}{Akidau et~al\mbox{.}}{2015}]%
        {journals/pvldb/AkidauBCCFLMMPS15}
\bibfield{author}{\bibinfo{person}{Tyler Akidau}, \bibinfo{person}{Robert
  Bradshaw}, \bibinfo{person}{Craig Chambers}, \bibinfo{person}{Slava
  Chernyak}, \bibinfo{person}{Rafael Fern{\'{a}}ndez{-}Moctezuma},
  \bibinfo{person}{Reuven Lax}, \bibinfo{person}{Sam McVeety},
  \bibinfo{person}{Daniel Mills}, \bibinfo{person}{Frances Perry},
  \bibinfo{person}{Eric Schmidt}, {and} \bibinfo{person}{Sam Whittle}.}
  \bibinfo{year}{2015}\natexlab{}.
\newblock \showarticletitle{The Dataflow Model: {A} Practical Approach to
  Balancing Correctness, Latency, and Cost in Massive-Scale, Unbounded,
  Out-of-Order Data Processing}.
\newblock \bibinfo{journal}{\emph{Proc. {VLDB} Endow.}} \bibinfo{volume}{8},
  \bibinfo{number}{12} (\bibinfo{year}{2015}), \bibinfo{pages}{1792--1803}.
\newblock
\urldef\tempurl%
\url{https://doi.org/10.14778/2824032.2824076}
\showDOI{\tempurl}


\bibitem[\protect\citeauthoryear{Armbrust, Das, Torres, Yavuz, Zhu, Xin,
  Ghodsi, Stoica, and Zaharia}{Armbrust et~al\mbox{.}}{2018}]%
        {conf/sigmod/ArmbrustDTYZX0S18}
\bibfield{author}{\bibinfo{person}{Michael Armbrust},
  \bibinfo{person}{Tathagata Das}, \bibinfo{person}{Joseph Torres},
  \bibinfo{person}{Burak Yavuz}, \bibinfo{person}{Shixiong Zhu},
  \bibinfo{person}{Reynold Xin}, \bibinfo{person}{Ali Ghodsi},
  \bibinfo{person}{Ion Stoica}, {and} \bibinfo{person}{Matei Zaharia}.}
  \bibinfo{year}{2018}\natexlab{}.
\newblock \showarticletitle{Structured Streaming: {A} Declarative {API} for
  Real-Time Applications in Apache Spark}. In
  \bibinfo{booktitle}{\emph{Proceedings of the 2018 International Conference on
  Management of Data, {SIGMOD} Conference 2018, Houston, TX, USA, June 10-15,
  2018}}, \bibfield{editor}{\bibinfo{person}{Gautam Das},
  \bibinfo{person}{Christopher~M. Jermaine}, {and} \bibinfo{person}{Philip~A.
  Bernstein}} (Eds.). \bibinfo{publisher}{{ACM}}, \bibinfo{pages}{601--613}.
\newblock
\urldef\tempurl%
\url{https://doi.org/10.1145/3183713.3190664}
\showDOI{\tempurl}


\bibitem[\protect\citeauthoryear{Bernstein, Hadzilacos, and Goodman}{Bernstein
  et~al\mbox{.}}{1987}]%
        {books/aw/BernsteinHG87}
\bibfield{author}{\bibinfo{person}{Philip~A. Bernstein},
  \bibinfo{person}{Vassos Hadzilacos}, {and} \bibinfo{person}{Nathan Goodman}.}
  \bibinfo{year}{1987}\natexlab{}.
\newblock \bibinfo{booktitle}{\emph{Concurrency Control and Recovery in
  Database Systems}}.
\newblock \bibinfo{publisher}{Addison-Wesley}.
\newblock
\showISBNx{0-201-10715-5}
\urldef\tempurl%
\url{http://research.microsoft.com/en-us/people/philbe/ccontrol.aspx}
\showURL{%
\tempurl}


\bibitem[\protect\citeauthoryear{Bernstein and Newcomer}{Bernstein and
  Newcomer}{1996}]%
        {books/mk/BernsteinN96}
\bibfield{author}{\bibinfo{person}{Philip~A. Bernstein} {and}
  \bibinfo{person}{Eric Newcomer}.} \bibinfo{year}{1996}\natexlab{}.
\newblock \bibinfo{booktitle}{\emph{Principles of Transaction Processing for
  Systems Professionals}}.
\newblock \bibinfo{publisher}{Morgan Kaufmann}.
\newblock
\showISBNx{1-55860-415-4}


\bibitem[\protect\citeauthoryear{Botan, Fischer, Kossmann, and Tatbul}{Botan
  et~al\mbox{.}}{2012}]%
        {conf/edbt/BotanFKT12}
\bibfield{author}{\bibinfo{person}{Irina Botan}, \bibinfo{person}{Peter~M.
  Fischer}, \bibinfo{person}{Donald Kossmann}, {and} \bibinfo{person}{Nesime
  Tatbul}.} \bibinfo{year}{2012}\natexlab{}.
\newblock \showarticletitle{Transactional stream processing}. In
  \bibinfo{booktitle}{\emph{15th International Conference on Extending Database
  Technology, {EDBT} '12, Berlin, Germany, March 27-30, 2012, Proceedings}},
  \bibfield{editor}{\bibinfo{person}{Elke~A. Rundensteiner},
  \bibinfo{person}{Volker Markl}, \bibinfo{person}{Ioana Manolescu},
  \bibinfo{person}{Sihem Amer{-}Yahia}, \bibinfo{person}{Felix Naumann}, {and}
  \bibinfo{person}{Ismail Ari}} (Eds.). \bibinfo{publisher}{{ACM}},
  \bibinfo{pages}{204--215}.
\newblock
\urldef\tempurl%
\url{https://doi.org/10.1145/2247596.2247622}
\showDOI{\tempurl}


\bibitem[\protect\citeauthoryear{Carbone, Ewen, F{\'{o}}ra, Haridi, Richter,
  and Tzoumas}{Carbone et~al\mbox{.}}{2017}]%
        {journals/pvldb/CarboneEFHRT17}
\bibfield{author}{\bibinfo{person}{Paris Carbone}, \bibinfo{person}{Stephan
  Ewen}, \bibinfo{person}{Gyula F{\'{o}}ra}, \bibinfo{person}{Seif Haridi},
  \bibinfo{person}{Stefan Richter}, {and} \bibinfo{person}{Kostas Tzoumas}.}
  \bibinfo{year}{2017}\natexlab{}.
\newblock \showarticletitle{State Management in Apache Flink{\textregistered}:
  Consistent Stateful Distributed Stream Processing}.
\newblock \bibinfo{journal}{\emph{Proc. {VLDB} Endow.}} \bibinfo{volume}{10},
  \bibinfo{number}{12} (\bibinfo{year}{2017}), \bibinfo{pages}{1718--1729}.
\newblock
\urldef\tempurl%
\url{https://doi.org/10.14778/3137765.3137777}
\showDOI{\tempurl}


\bibitem[\protect\citeauthoryear{Carbone, F{\'{o}}ra, Ewen, Haridi, and
  Tzoumas}{Carbone et~al\mbox{.}}{2015a}]%
        {journals/corr/CarboneFEHT15}
\bibfield{author}{\bibinfo{person}{Paris Carbone}, \bibinfo{person}{Gyula
  F{\'{o}}ra}, \bibinfo{person}{Stephan Ewen}, \bibinfo{person}{Seif Haridi},
  {and} \bibinfo{person}{Kostas Tzoumas}.} \bibinfo{year}{2015}\natexlab{a}.
\newblock \showarticletitle{Lightweight Asynchronous Snapshots for Distributed
  Dataflows}.
\newblock \bibinfo{journal}{\emph{CoRR}}  \bibinfo{volume}{abs/1506.08603}
  (\bibinfo{year}{2015}).
\newblock
\showeprint[arXiv]{1506.08603}
\urldef\tempurl%
\url{http://arxiv.org/abs/1506.08603}
\showURL{%
\tempurl}


\bibitem[\protect\citeauthoryear{Carbone, Fragkoulis, Kalavri, and
  Katsifodimos}{Carbone et~al\mbox{.}}{2020}]%
        {conf/sigmod/CarboneFKK20}
\bibfield{author}{\bibinfo{person}{Paris Carbone}, \bibinfo{person}{Marios
  Fragkoulis}, \bibinfo{person}{Vasiliki Kalavri}, {and}
  \bibinfo{person}{Asterios Katsifodimos}.} \bibinfo{year}{2020}\natexlab{}.
\newblock \showarticletitle{Beyond Analytics: The Evolution of Stream
  Processing Systems}. In \bibinfo{booktitle}{\emph{Proceedings of the 2020
  International Conference on Management of Data, {SIGMOD} Conference 2020,
  online conference [Portland, OR, USA], June 14-19, 2020}},
  \bibfield{editor}{\bibinfo{person}{David Maier}, \bibinfo{person}{Rachel
  Pottinger}, \bibinfo{person}{AnHai Doan}, \bibinfo{person}{Wang{-}Chiew Tan},
  \bibinfo{person}{Abdussalam Alawini}, {and} \bibinfo{person}{Hung~Q. Ngo}}
  (Eds.). \bibinfo{publisher}{{ACM}}, \bibinfo{pages}{2651--2658}.
\newblock
\urldef\tempurl%
\url{https://doi.org/10.1145/3318464.3383131}
\showDOI{\tempurl}


\bibitem[\protect\citeauthoryear{Carbone, Katsifodimos, Ewen, Markl, Haridi,
  and Tzoumas}{Carbone et~al\mbox{.}}{2015b}]%
        {journals/debu/CarboneKEMHT15}
\bibfield{author}{\bibinfo{person}{Paris Carbone}, \bibinfo{person}{Asterios
  Katsifodimos}, \bibinfo{person}{Stephan Ewen}, \bibinfo{person}{Volker
  Markl}, \bibinfo{person}{Seif Haridi}, {and} \bibinfo{person}{Kostas
  Tzoumas}.} \bibinfo{year}{2015}\natexlab{b}.
\newblock \showarticletitle{Apache Flink{\texttrademark}: Stream and Batch
  Processing in a Single Engine}.
\newblock \bibinfo{journal}{\emph{{IEEE} Data Eng. Bull.}}
  \bibinfo{volume}{38}, \bibinfo{number}{4} (\bibinfo{year}{2015}),
  \bibinfo{pages}{28--38}.
\newblock
\urldef\tempurl%
\url{http://sites.computer.org/debull/A15dec/p28.pdf}
\showURL{%
\tempurl}


\bibitem[\protect\citeauthoryear{Chandramouli, Goldstein, Barnett, DeLine,
  Platt, Terwilliger, and Wernsing}{Chandramouli et~al\mbox{.}}{2014}]%
        {journals/pvldb/ChandramouliGBDPTW14}
\bibfield{author}{\bibinfo{person}{Badrish Chandramouli},
  \bibinfo{person}{Jonathan Goldstein}, \bibinfo{person}{Mike Barnett},
  \bibinfo{person}{Robert DeLine}, \bibinfo{person}{John~C. Platt},
  \bibinfo{person}{James~F. Terwilliger}, {and} \bibinfo{person}{John
  Wernsing}.} \bibinfo{year}{2014}\natexlab{}.
\newblock \showarticletitle{Trill: {A} High-Performance Incremental Query
  Processor for Diverse Analytics}.
\newblock \bibinfo{journal}{\emph{Proc. {VLDB} Endow.}} \bibinfo{volume}{8},
  \bibinfo{number}{4} (\bibinfo{year}{2014}), \bibinfo{pages}{401--412}.
\newblock
\urldef\tempurl%
\url{https://doi.org/10.14778/2735496.2735503}
\showDOI{\tempurl}


\bibitem[\protect\citeauthoryear{Cormen}{Cormen}{2013}]%
        {books/daglib/0032835}
\bibfield{author}{\bibinfo{person}{Thomas~H. Cormen}.}
  \bibinfo{year}{2013}\natexlab{}.
\newblock \bibinfo{booktitle}{\emph{Algorithms Unlocked}}.
\newblock \bibinfo{publisher}{{MIT} Press}.
\newblock
\showISBNx{978-0-262-51880-2}
\urldef\tempurl%
\url{http://mitpress.mit.edu/books/algorithms-unlocked}
\showURL{%
\tempurl}


\bibitem[\protect\citeauthoryear{Dasgupta, Papadimitriou, and
  Vazirani}{Dasgupta et~al\mbox{.}}{2008}]%
        {books/daglib/0017733}
\bibfield{author}{\bibinfo{person}{Sanjoy Dasgupta},
  \bibinfo{person}{Christos~H. Papadimitriou}, {and} \bibinfo{person}{Umesh~V.
  Vazirani}.} \bibinfo{year}{2008}\natexlab{}.
\newblock \bibinfo{booktitle}{\emph{Algorithms}}.
\newblock \bibinfo{publisher}{McGraw-Hill}.
\newblock
\showISBNx{978-0-07-352340-8}


\bibitem[\protect\citeauthoryear{Fernandez, Migliavacca, Kalyvianaki, and
  Pietzuch}{Fernandez et~al\mbox{.}}{2013}]%
        {conf/sigmod/FernandezMKP13}
\bibfield{author}{\bibinfo{person}{Raul~Castro Fernandez},
  \bibinfo{person}{Matteo Migliavacca}, \bibinfo{person}{Evangelia
  Kalyvianaki}, {and} \bibinfo{person}{Peter~R. Pietzuch}.}
  \bibinfo{year}{2013}\natexlab{}.
\newblock \showarticletitle{Integrating scale out and fault tolerance in stream
  processing using operator state management}. In
  \bibinfo{booktitle}{\emph{Proceedings of the {ACM} {SIGMOD} International
  Conference on Management of Data, {SIGMOD} 2013, New York, NY, USA, June
  22-27, 2013}}, \bibfield{editor}{\bibinfo{person}{Kenneth~A. Ross},
  \bibinfo{person}{Divesh Srivastava}, {and} \bibinfo{person}{Dimitris
  Papadias}} (Eds.). \bibinfo{publisher}{{ACM}}, \bibinfo{pages}{725--736}.
\newblock
\urldef\tempurl%
\url{https://doi.org/10.1145/2463676.2465282}
\showDOI{\tempurl}


\bibitem[\protect\citeauthoryear{FlinkFraudDetectionDemo}{FlinkFraudDetectionDemo}{[n.d.]}]%
        {FlinkFraudDetectionDemo}
FlinkFraudDetectionDemo \bibinfo{year}{[n.d.]}\natexlab{}.
\newblock
\newblock
\newblock
\shownote{Advanced Flink Application Patterns Vol.1: Case Study of a Fraud
  Detection System,
  \url{https://flink.apache.org/news/2020/01/15/demo-fraud-detection.html}.}


\bibitem[\protect\citeauthoryear{FlinkSavepoint}{FlinkSavepoint}{[n.d.]}]%
        {FlinkSavepoint}
FlinkSavepoint \bibinfo{year}{[n.d.]}\natexlab{}.
\newblock
\newblock
\newblock
\shownote{Savepoints in Apache Flink,
  \url{https://ci.apache.org/projects/flink/flink-docs-master/docs/ops/state/savepoints/}.}


\bibitem[\protect\citeauthoryear{FlinkUpdateCepPattern}{FlinkUpdateCepPattern}{[n.d.]}]%
        {FlinkUpdateCepPattern}
FlinkUpdateCepPattern \bibinfo{year}{[n.d.]}\natexlab{}.
\newblock
\newblock
\newblock
\shownote{Support dynamically changing CEP patterns in Flink,
  \url{https://issues.apache.org/jira/browse/FLINK-7129}.}


\bibitem[\protect\citeauthoryear{FlinkUpdateVol2}{FlinkUpdateVol2}{[n.d.]}]%
        {FlinkUpdateLogicVol2}
FlinkUpdateVol2 \bibinfo{year}{[n.d.]}\natexlab{}.
\newblock
\newblock
\newblock
\shownote{Advanced Flink Application Patterns Vol.2: Dynamic Updates of
  Application Logic,
  \url{https://flink.apache.org/news/2020/03/24/demo-fraud-detection-2.html}.}


\bibitem[\protect\citeauthoryear{Gjengset, Schwarzkopf, Behrens, Ara{\'{u}}jo,
  Ek, Kohler, Kaashoek, and Morris}{Gjengset et~al\mbox{.}}{2018}]%
        {conf/osdi/GjengsetSBAEKKM18}
\bibfield{author}{\bibinfo{person}{Jon Gjengset}, \bibinfo{person}{Malte
  Schwarzkopf}, \bibinfo{person}{Jonathan Behrens},
  \bibinfo{person}{Lara~Timb{\'{o}} Ara{\'{u}}jo}, \bibinfo{person}{Martin Ek},
  \bibinfo{person}{Eddie Kohler}, \bibinfo{person}{M.~Frans Kaashoek}, {and}
  \bibinfo{person}{Robert~Tappan Morris}.} \bibinfo{year}{2018}\natexlab{}.
\newblock \showarticletitle{Noria: dynamic, partially-stateful data-flow for
  high-performance web applications}. In \bibinfo{booktitle}{\emph{13th
  {USENIX} Symposium on Operating Systems Design and Implementation, {OSDI}
  2018, Carlsbad, CA, USA, October 8-10, 2018}},
  \bibfield{editor}{\bibinfo{person}{Andrea~C. Arpaci{-}Dusseau} {and}
  \bibinfo{person}{Geoff Voelker}} (Eds.). \bibinfo{publisher}{{USENIX}
  Association}, \bibinfo{pages}{213--231}.
\newblock
\urldef\tempurl%
\url{https://www.usenix.org/conference/osdi18/presentation/gjengset}
\showURL{%
\tempurl}


\bibitem[\protect\citeauthoryear{Hoffmann, Lattuada, McSherry, Kalavri,
  Liagouris, and Roscoe}{Hoffmann et~al\mbox{.}}{2019}]%
        {journals/pvldb/HoffmannLMKLR19}
\bibfield{author}{\bibinfo{person}{Moritz Hoffmann}, \bibinfo{person}{Andrea
  Lattuada}, \bibinfo{person}{Frank McSherry}, \bibinfo{person}{Vasiliki
  Kalavri}, \bibinfo{person}{John Liagouris}, {and} \bibinfo{person}{Timothy
  Roscoe}.} \bibinfo{year}{2019}\natexlab{}.
\newblock \showarticletitle{Megaphone: Latency-conscious state migration for
  distributed streaming dataflows}.
\newblock \bibinfo{journal}{\emph{Proc. {VLDB} Endow.}} \bibinfo{volume}{12},
  \bibinfo{number}{9} (\bibinfo{year}{2019}), \bibinfo{pages}{1002--1015}.
\newblock
\urldef\tempurl%
\url{https://doi.org/10.14778/3329772.3329777}
\showDOI{\tempurl}


\bibitem[\protect\citeauthoryear{Kakousis, Paspallis, and
  Papadopoulos}{Kakousis et~al\mbox{.}}{2010}]%
        {journals/eis/KakousisPP10}
\bibfield{author}{\bibinfo{person}{Konstantinos Kakousis},
  \bibinfo{person}{Nearchos Paspallis}, {and} \bibinfo{person}{George~Angelos
  Papadopoulos}.} \bibinfo{year}{2010}\natexlab{}.
\newblock \showarticletitle{A survey of software adaptation in mobile and
  ubiquitous computing}.
\newblock \bibinfo{journal}{\emph{Enterp. Inf. Syst.}} \bibinfo{volume}{4},
  \bibinfo{number}{4} (\bibinfo{year}{2010}), \bibinfo{pages}{355--389}.
\newblock
\urldef\tempurl%
\url{https://doi.org/10.1080/17517575.2010.509814}
\showDOI{\tempurl}


\bibitem[\protect\citeauthoryear{Kon and Campbell}{Kon and Campbell}{1999}]%
        {conf/coots/KonC99}
\bibfield{author}{\bibinfo{person}{Fabio Kon} {and} \bibinfo{person}{Roy~H.
  Campbell}.} \bibinfo{year}{1999}\natexlab{}.
\newblock \showarticletitle{Supporting Automatic Configuration of
  Component-Based Distributed Systems}. In
  \bibinfo{booktitle}{\emph{Proceedings of the 5th {USENIX} Conference on
  Object-Oriented Technologies {\&} Systems, May 3-7, 1999, The Town {\&}
  Country Resort Hotel, San Diego, California, {USA}}},
  \bibfield{editor}{\bibinfo{person}{Murthy~V. Devarakonda}} (Ed.).
  \bibinfo{publisher}{{USENIX}}, \bibinfo{pages}{175--188}.
\newblock
\urldef\tempurl%
\url{http://www.usenix.org/publications/library/proceedings/coots99/kon.html}
\showURL{%
\tempurl}


\bibitem[\protect\citeauthoryear{Kumar, Wang, Ni, and Li}{Kumar
  et~al\mbox{.}}{2020}]%
        {journals/pvldb/KumarWNL20}
\bibfield{author}{\bibinfo{person}{Avinash Kumar}, \bibinfo{person}{Zuozhi
  Wang}, \bibinfo{person}{Shengquan Ni}, {and} \bibinfo{person}{Chen Li}.}
  \bibinfo{year}{2020}\natexlab{}.
\newblock \showarticletitle{Amber: {A} Debuggable Dataflow System Based on the
  Actor Model}.
\newblock \bibinfo{journal}{\emph{Proc. {VLDB} Endow.}} \bibinfo{volume}{13},
  \bibinfo{number}{5} (\bibinfo{year}{2020}), \bibinfo{pages}{740--753}.
\newblock
\urldef\tempurl%
\url{https://doi.org/10.14778/3377369.3377381}
\showDOI{\tempurl}


\bibitem[\protect\citeauthoryear{Ma, Baresi, Ghezzi, Manna, and Lu}{Ma
  et~al\mbox{.}}{2011}]%
        {conf/sigsoft/MaBGML11}
\bibfield{author}{\bibinfo{person}{Xiaoxing Ma}, \bibinfo{person}{Luciano
  Baresi}, \bibinfo{person}{Carlo Ghezzi}, \bibinfo{person}{Valerio Panzica~La
  Manna}, {and} \bibinfo{person}{Jian Lu}.} \bibinfo{year}{2011}\natexlab{}.
\newblock \showarticletitle{Version-consistent dynamic reconfiguration of
  component-based distributed systems}. In
  \bibinfo{booktitle}{\emph{SIGSOFT/FSE'11 19th {ACM} {SIGSOFT} Symposium on
  the Foundations of Software Engineering {(FSE-19)} and ESEC'11: 13th European
  Software Engineering Conference (ESEC-13), Szeged, Hungary, September 5-9,
  2011}}, \bibfield{editor}{\bibinfo{person}{Tibor Gyim{\'{o}}thy} {and}
  \bibinfo{person}{Andreas Zeller}} (Eds.). \bibinfo{publisher}{{ACM}},
  \bibinfo{pages}{245--255}.
\newblock
\urldef\tempurl%
\url{https://doi.org/10.1145/2025113.2025148}
\showDOI{\tempurl}


\bibitem[\protect\citeauthoryear{Mai, Zeng, Potharaju, Xu, Suh, Venkataraman,
  Costa, Kim, Muthukrishnan, Kuppa, Dhulipalla, and Rao}{Mai
  et~al\mbox{.}}{2018}]%
        {journals/pvldb/MaiZPXSVCKMKDR18}
\bibfield{author}{\bibinfo{person}{Luo Mai}, \bibinfo{person}{Kai Zeng},
  \bibinfo{person}{Rahul Potharaju}, \bibinfo{person}{Le Xu},
  \bibinfo{person}{Steve Suh}, \bibinfo{person}{Shivaram Venkataraman},
  \bibinfo{person}{Paolo Costa}, \bibinfo{person}{Terry Kim},
  \bibinfo{person}{Saravanam Muthukrishnan}, \bibinfo{person}{Vamsi Kuppa},
  \bibinfo{person}{Sudheer Dhulipalla}, {and} \bibinfo{person}{Sriram Rao}.}
  \bibinfo{year}{2018}\natexlab{}.
\newblock \showarticletitle{Chi: {A} Scalable and Programmable Control Plane
  for Distributed Stream Processing Systems}.
\newblock \bibinfo{journal}{\emph{Proc. {VLDB} Endow.}} \bibinfo{volume}{11},
  \bibinfo{number}{10} (\bibinfo{year}{2018}), \bibinfo{pages}{1303--1316}.
\newblock
\urldef\tempurl%
\url{https://doi.org/10.14778/3231751.3231765}
\showDOI{\tempurl}


\bibitem[\protect\citeauthoryear{Mao, Huang, Tian, Wang, and Ma}{Mao
  et~al\mbox{.}}{2021}]%
        {conf/cloud/MaoHTWM21}
\bibfield{author}{\bibinfo{person}{Yancan Mao}, \bibinfo{person}{Yuan Huang},
  \bibinfo{person}{Runxin Tian}, \bibinfo{person}{Xin Wang}, {and}
  \bibinfo{person}{Richard T.~B. Ma}.} \bibinfo{year}{2021}\natexlab{}.
\newblock \showarticletitle{Trisk: Task-Centric Data Stream Reconfiguration}.
  In \bibinfo{booktitle}{\emph{SoCC '21: {ACM} Symposium on Cloud Computing,
  Seattle, WA, USA, November 1 - 4, 2021}},
  \bibfield{editor}{\bibinfo{person}{Carlo Curino}, \bibinfo{person}{Georgia
  Koutrika}, {and} \bibinfo{person}{Ravi Netravali}} (Eds.).
  \bibinfo{publisher}{{ACM}}, \bibinfo{pages}{214--228}.
\newblock
\urldef\tempurl%
\url{https://doi.org/10.1145/3472883.3487010}
\showDOI{\tempurl}


\bibitem[\protect\citeauthoryear{Meehan, Tatbul, Zdonik, Aslantas,
  {\c{C}}etintemel, Du, Kraska, Madden, Maier, Pavlo, Stonebraker, Tufte, and
  Wang}{Meehan et~al\mbox{.}}{2015}]%
        {journals/pvldb/MeehanTZACDKMMP15}
\bibfield{author}{\bibinfo{person}{John Meehan}, \bibinfo{person}{Nesime
  Tatbul}, \bibinfo{person}{Stan Zdonik}, \bibinfo{person}{Cansu Aslantas},
  \bibinfo{person}{Ugur {\c{C}}etintemel}, \bibinfo{person}{Jiang Du},
  \bibinfo{person}{Tim Kraska}, \bibinfo{person}{Samuel Madden},
  \bibinfo{person}{David Maier}, \bibinfo{person}{Andrew Pavlo},
  \bibinfo{person}{Michael Stonebraker}, \bibinfo{person}{Kristin Tufte}, {and}
  \bibinfo{person}{Hao Wang}.} \bibinfo{year}{2015}\natexlab{}.
\newblock \showarticletitle{S-Store: Streaming Meets Transaction Processing}.
\newblock \bibinfo{journal}{\emph{Proc. {VLDB} Endow.}} \bibinfo{volume}{8},
  \bibinfo{number}{13} (\bibinfo{year}{2015}), \bibinfo{pages}{2134--2145}.
\newblock
\urldef\tempurl%
\url{https://doi.org/10.14778/2831360.2831367}
\showDOI{\tempurl}


\bibitem[\protect\citeauthoryear{Monte, Zeuch, Rabl, and Markl}{Monte
  et~al\mbox{.}}{2020}]%
        {conf/sigmod/MonteZRM20}
\bibfield{author}{\bibinfo{person}{Bonaventura~Del Monte},
  \bibinfo{person}{Steffen Zeuch}, \bibinfo{person}{Tilmann Rabl}, {and}
  \bibinfo{person}{Volker Markl}.} \bibinfo{year}{2020}\natexlab{}.
\newblock \showarticletitle{Rhino: Efficient Management of Very Large
  Distributed State for Stream Processing Engines}. In
  \bibinfo{booktitle}{\emph{Proceedings of the 2020 International Conference on
  Management of Data, {SIGMOD} Conference 2020, online conference [Portland,
  OR, USA], June 14-19, 2020}}, \bibfield{editor}{\bibinfo{person}{David
  Maier}, \bibinfo{person}{Rachel Pottinger}, \bibinfo{person}{AnHai Doan},
  \bibinfo{person}{Wang{-}Chiew Tan}, \bibinfo{person}{Abdussalam Alawini},
  {and} \bibinfo{person}{Hung~Q. Ngo}} (Eds.). \bibinfo{publisher}{{ACM}},
  \bibinfo{pages}{2471--2486}.
\newblock
\urldef\tempurl%
\url{https://doi.org/10.1145/3318464.3389723}
\showDOI{\tempurl}


\bibitem[\protect\citeauthoryear{Murray, McSherry, Isaacs, Isard, Barham, and
  Abadi}{Murray et~al\mbox{.}}{2013}]%
        {conf/sosp/MurrayMIIBA13}
\bibfield{author}{\bibinfo{person}{Derek~Gordon Murray}, \bibinfo{person}{Frank
  McSherry}, \bibinfo{person}{Rebecca Isaacs}, \bibinfo{person}{Michael Isard},
  \bibinfo{person}{Paul Barham}, {and} \bibinfo{person}{Mart{\'{\i}}n Abadi}.}
  \bibinfo{year}{2013}\natexlab{}.
\newblock \showarticletitle{Naiad: a timely dataflow system}. In
  \bibinfo{booktitle}{\emph{{ACM} {SIGOPS} 24th Symposium on Operating Systems
  Principles, {SOSP} '13, Farmington, PA, USA, November 3-6, 2013}},
  \bibfield{editor}{\bibinfo{person}{Michael Kaminsky} {and}
  \bibinfo{person}{Mike Dahlin}} (Eds.). \bibinfo{publisher}{{ACM}},
  \bibinfo{pages}{439--455}.
\newblock
\urldef\tempurl%
\url{https://doi.org/10.1145/2517349.2522738}
\showDOI{\tempurl}


\bibitem[\protect\citeauthoryear{Padhi, Schiff, Melnyk, Rigotti, Mroueh,
  Dognin, Ross, Nair, and Altman}{Padhi et~al\mbox{.}}{2021}]%
        {padhi2021tabular}
\bibfield{author}{\bibinfo{person}{Inkit Padhi}, \bibinfo{person}{Yair Schiff},
  \bibinfo{person}{Igor Melnyk}, \bibinfo{person}{Mattia Rigotti},
  \bibinfo{person}{Youssef Mroueh}, \bibinfo{person}{Pierre Dognin},
  \bibinfo{person}{Jerret Ross}, \bibinfo{person}{Ravi Nair}, {and}
  \bibinfo{person}{Erik Altman}.} \bibinfo{year}{2021}\natexlab{}.
\newblock \showarticletitle{Tabular transformers for modeling multivariate time
  series}. In \bibinfo{booktitle}{\emph{ICASSP 2021-2021 IEEE International
  Conference on Acoustics, Speech and Signal Processing (ICASSP)}}. IEEE,
  \bibinfo{pages}{3565--3569}.
\newblock
\urldef\tempurl%
\url{https://ieeexplore.ieee.org/document/9414142}
\showURL{%
\tempurl}


\bibitem[\protect\citeauthoryear{Sadeghi, Esfahani, and Malek}{Sadeghi
  et~al\mbox{.}}{2017}]%
        {journals/tosem/SadeghiEM17}
\bibfield{author}{\bibinfo{person}{Alireza Sadeghi}, \bibinfo{person}{Naeem
  Esfahani}, {and} \bibinfo{person}{Sam Malek}.}
  \bibinfo{year}{2017}\natexlab{}.
\newblock \showarticletitle{Ensuring the Consistency of Adaptation through
  Inter- and Intra-Component Dependency Analysis}.
\newblock \bibinfo{journal}{\emph{{ACM} Trans. Softw. Eng. Methodol.}}
  \bibinfo{volume}{26}, \bibinfo{number}{1} (\bibinfo{year}{2017}),
  \bibinfo{pages}{2:1--2:27}.
\newblock
\urldef\tempurl%
\url{https://doi.org/10.1145/3063385}
\showDOI{\tempurl}


\bibitem[\protect\citeauthoryear{Silvestre, Fragkoulis, Spinellis, and
  Katsifodimos}{Silvestre et~al\mbox{.}}{2021}]%
        {conf/sigmod/SilvestreFSK21}
\bibfield{author}{\bibinfo{person}{Pedro~F. Silvestre}, \bibinfo{person}{Marios
  Fragkoulis}, \bibinfo{person}{Diomidis Spinellis}, {and}
  \bibinfo{person}{Asterios Katsifodimos}.} \bibinfo{year}{2021}\natexlab{}.
\newblock \showarticletitle{Clonos: Consistent Causal Recovery for
  Highly-Available Streaming Dataflows}. In \bibinfo{booktitle}{\emph{{SIGMOD}
  '21: International Conference on Management of Data, Virtual Event, China,
  June 20-25, 2021}}, \bibfield{editor}{\bibinfo{person}{Guoliang Li},
  \bibinfo{person}{Zhanhuai Li}, \bibinfo{person}{Stratos Idreos}, {and}
  \bibinfo{person}{Divesh Srivastava}} (Eds.). \bibinfo{publisher}{{ACM}},
  \bibinfo{pages}{1637--1650}.
\newblock
\urldef\tempurl%
\url{https://doi.org/10.1145/3448016.3457320}
\showDOI{\tempurl}


\bibitem[\protect\citeauthoryear{StreamINGFraudDetection}{StreamINGFraudDetection}{[n.d.]}]%
        {StreamINGFraudDetection}
StreamINGFraudDetection \bibinfo{year}{[n.d.]}\natexlab{}.
\newblock
\newblock
\newblock
\shownote{StreamING Machine Learning Models: How ING Adds Fraud Detection
  Models at Runtime with Apache Flink,
  \url{https://www.ververica.com/blog/real-time-fraud-detection-ing-bank-apache-flink}.}


\bibitem[\protect\citeauthoryear{Tiwari, Ramprasad, Mortazavi, Gabel, and
  de~Lara}{Tiwari et~al\mbox{.}}{2019}]%
        {conf/wmcsa/TiwariRMGL19}
\bibfield{author}{\bibinfo{person}{Abhishek Tiwari}, \bibinfo{person}{Brian
  Ramprasad}, \bibinfo{person}{Seyed~Hossein Mortazavi}, \bibinfo{person}{Moshe
  Gabel}, {and} \bibinfo{person}{Eyal de Lara}.}
  \bibinfo{year}{2019}\natexlab{}.
\newblock \showarticletitle{Reconfigurable Streaming for the Mobile Edge}. In
  \bibinfo{booktitle}{\emph{Proceedings of the 20th International Workshop on
  Mobile Computing Systems and Applications, HotMobile 2019, Santa Cruz, CA,
  USA, February 27-28, 2019}}, \bibfield{editor}{\bibinfo{person}{Alec Wolman}
  {and} \bibinfo{person}{Lin Zhong}} (Eds.). \bibinfo{publisher}{{ACM}},
  \bibinfo{pages}{153--158}.
\newblock
\urldef\tempurl%
\url{https://doi.org/10.1145/3301293.3302355}
\showDOI{\tempurl}


\bibitem[\protect\citeauthoryear{Toshniwal, Taneja, Shukla, Ramasamy, Patel,
  Kulkarni, Jackson, Gade, Fu, Donham, Bhagat, Mittal, and Ryaboy}{Toshniwal
  et~al\mbox{.}}{2014}]%
        {conf/sigmod/ToshniwalTSRPKJGFDBMR14}
\bibfield{author}{\bibinfo{person}{Ankit Toshniwal}, \bibinfo{person}{Siddarth
  Taneja}, \bibinfo{person}{Amit Shukla}, \bibinfo{person}{Karthikeyan
  Ramasamy}, \bibinfo{person}{Jignesh~M. Patel}, \bibinfo{person}{Sanjeev
  Kulkarni}, \bibinfo{person}{Jason Jackson}, \bibinfo{person}{Krishna Gade},
  \bibinfo{person}{Maosong Fu}, \bibinfo{person}{Jake Donham},
  \bibinfo{person}{Nikunj Bhagat}, \bibinfo{person}{Sailesh Mittal}, {and}
  \bibinfo{person}{Dmitriy~V. Ryaboy}.} \bibinfo{year}{2014}\natexlab{}.
\newblock \showarticletitle{Storm@twitter}. In
  \bibinfo{booktitle}{\emph{International Conference on Management of Data,
  {SIGMOD} 2014, Snowbird, UT, USA, June 22-27, 2014}},
  \bibfield{editor}{\bibinfo{person}{Curtis~E. Dyreson},
  \bibinfo{person}{Feifei Li}, {and} \bibinfo{person}{M.~Tamer {\"{O}}zsu}}
  (Eds.). \bibinfo{publisher}{{ACM}}, \bibinfo{pages}{147--156}.
\newblock
\urldef\tempurl%
\url{https://doi.org/10.1145/2588555.2595641}
\showDOI{\tempurl}


\bibitem[\protect\citeauthoryear{TPC-DS}{TPC-DS}{[n.d.]}]%
        {misc/tpcds}
TPC-DS \bibinfo{year}{[n.d.]}\natexlab{}.
\newblock
\newblock
\newblock
\shownote{TPC-DS http://www.tpc.org/tpcds/.}


\bibitem[\protect\citeauthoryear{UpgradeFlinkApplications}{UpgradeFlinkApplications}{[n.d.]}]%
        {UpgradeFlinkApplications}
UpgradeFlinkApplications \bibinfo{year}{[n.d.]}\natexlab{}.
\newblock
\newblock
\newblock
\shownote{Upgrading Applications and Flink Versions,
  \url{https://nightlies.apache.org/flink/flink-docs-release-1.14/docs/ops/upgrading/}.}


\bibitem[\protect\citeauthoryear{Weikum and Vossen}{Weikum and Vossen}{2002}]%
        {books/mk/WeikumV2002}
\bibfield{author}{\bibinfo{person}{Gerhard Weikum} {and}
  \bibinfo{person}{Gottfried Vossen}.} \bibinfo{year}{2002}\natexlab{}.
\newblock \bibinfo{booktitle}{\emph{Transactional Information Systems: Theory,
  Algorithms, and the Practice of Concurrency Control and Recovery}}.
\newblock \bibinfo{publisher}{Morgan Kaufmann}.
\newblock
\showISBNx{1-55860-508-8}


\bibitem[\protect\citeauthoryear{Wiese and Omlin}{Wiese and Omlin}{2009}]%
        {wiese2009credit}
\bibfield{author}{\bibinfo{person}{B{\'e}nard Wiese} {and}
  \bibinfo{person}{Christian Omlin}.} \bibinfo{year}{2009}\natexlab{}.
\newblock \showarticletitle{Credit card transactions, fraud detection, and
  machine learning: Modelling time with LSTM recurrent neural networks}.
\newblock In \bibinfo{booktitle}{\emph{Innovations in neural information
  paradigms and applications}}. \bibinfo{publisher}{Springer},
  \bibinfo{pages}{231--268}.
\newblock


\bibitem[\protect\citeauthoryear{Zaharia, Das, Li, Hunter, Shenker, and
  Stoica}{Zaharia et~al\mbox{.}}{2013}]%
        {conf/sosp/ZahariaDLHSS13}
\bibfield{author}{\bibinfo{person}{Matei Zaharia}, \bibinfo{person}{Tathagata
  Das}, \bibinfo{person}{Haoyuan Li}, \bibinfo{person}{Timothy Hunter},
  \bibinfo{person}{Scott Shenker}, {and} \bibinfo{person}{Ion Stoica}.}
  \bibinfo{year}{2013}\natexlab{}.
\newblock \showarticletitle{Discretized streams: fault-tolerant streaming
  computation at scale}. In \bibinfo{booktitle}{\emph{{ACM} {SIGOPS} 24th
  Symposium on Operating Systems Principles, {SOSP} '13, Farmington, PA, USA,
  November 3-6, 2013}}, \bibfield{editor}{\bibinfo{person}{Michael Kaminsky}
  {and} \bibinfo{person}{Mike Dahlin}} (Eds.). \bibinfo{publisher}{{ACM}},
  \bibinfo{pages}{423--438}.
\newblock
\urldef\tempurl%
\url{https://doi.org/10.1145/2517349.2522737}
\showDOI{\tempurl}


\end{thebibliography}

\end{document}